\documentclass[11pt,a4paper]{article}

\usepackage{graphicx}
\graphicspath{{prebuiltimages/}, {images/}} 

\usepackage{amssymb,amsthm} 
\usepackage{multicol} 
\usepackage{xcolor}
\usepackage[T1]{fontenc}
\usepackage[utf8]{inputenc}
\usepackage{lmodern}
\usepackage{natbib} 
\usepackage{geometry} 

\usepackage{color}

\usepackage{amsfonts}
\usepackage{multirow}
\usepackage{mathtools}

\usepackage{enumitem}
\usepackage{sidecap}
\usepackage{bm}
\usepackage{mathrsfs}
\usepackage{dsfont}
\usepackage{comment}

\usepackage{algorithm}
\usepackage{algorithmic}
\usepackage[titlenumbered,ruled,noend,algo2e]{algorithm2e}

\SetCommentSty{mycommfont}
\SetEndCharOfAlgoLine{}

\usepackage{numprint} 
\usepackage[toc,page,titletoc]{appendix}
\npthousandsep{\,}

\definecolor{my_link}{rgb}{0.64,0.16,0.16}
\definecolor{vertsombre}{rgb}{0.00,0.57,0.1}
\definecolor{marron}{rgb}{0.64,0.16,0.16}
\definecolor{orange_js}{RGB}{230,159,0}

\usepackage[naturalnames, bookmarks=true, plainpages=true, linkcolor=my_link, citecolor=marron, colorlinks=true, backref=page, urlcolor=blue]{hyperref}

\usepackage{./shortcuts_jac}
\crefname{supp}{Supplement}{Supplements}

\def\*#1{\mathbf{#1}}

\title{Spatially relaxed inference\\ on high-dimensional linear models}
\author{Jérôme-Alexis CHEVALIER \\
  Inria Paris-Saclay, CEA, Universit\'e Paris-Saclay\\
  and \\
  Tuan-Binh NGUYEN \\
  Inria Paris-Saclay, CEA, Universit\'e Paris-Saclay, LMO\\
  and \\
  Bertrand THIRION \\
  Inria Paris-Saclay, CEA, Universit\'e Paris-Saclay\\
  and \\
  Joseph SALMON \\
  IMAG, Université de Montpellier, CNRS\\
\\\href{mailto:jerome-alexis.chevalier@inria.fr}{jerome-alexis.chevalier@inria.fr}}

\date{\today}

\begin{document}
\newcounter{align}[equation] 

\maketitle

\begin{abstract}
We consider the inference problem for high-dimensional linear models,
when covariates have an underlying spatial organization reflected
in their correlation.
A typical example of such a setting is high-resolution imaging, in
which neighboring pixels are usually very similar.
Accurate point and confidence intervals estimation is not possible in
this context with many more covariates than samples, furthermore with
high correlation between  covariates.
This calls for a reformulation of the statistical inference problem,
that takes into account the underlying spatial structure: if covariates
are locally correlated, it is acceptable to detect them up to a given
spatial uncertainty.
We thus propose to rely on the $\delta$-FWER, that is the probability
of making a false discovery at a distance greater than $\delta$
from any true positive.
With this target measure in mind, we study the properties of
ensembled clustered inference algorithms which combine three
techniques: spatially constrained clustering, statistical inference,
and ensembling to aggregate several clustered inference solutions.
We show that ensembled clustered inference algorithms
control the $\delta$-FWER under standard assumptions
for $\delta$ equal to the largest cluster diameter.
We complement the theoretical analysis with empirical results,
demonstrating accurate $\delta$-FWER control and decent power achieved
by such inference algorithms.
\end{abstract}
\section{Introduction}
\label{sec:introduction_theory}
%
\paragraph{High-dimensional setting.}
High-dimensional regression corresponds to a setting where the
number of covariates (or features) $p$ exceeds the number of samples $n$.
It notably occurs when searching for conditional associations
among some high-dimensional observations and some outcome of interest:
the \emph{target}.
Typical examples of the high-dimensional setting include inference problems
on high-resolution images, where one aims at pixel- or voxel-level analysis,
\eg in neuroimaging \citep{norman2006, demartino2008}, astronomy
\citep{richards2009exploiting}, but also in other fields where
covariates display a spatial structure \eg in genomics
\citep{balding2006tutorial, dehman2015performance}.
In all these examples, it actually turns out that not only
$n < p$ but even $n\ll p$ and the covariates are spatially
structured because of the physics of the problem or the
measurements process.
Because such high-dimensional data lead to high-variance results,
probing statistical significance is important to give a level of
confidence in the reported association.
For this reason, the present analysis departs from traditional sparse
modeling methods such as the Lasso \citep{tibshirani1996}, that simply
aim at selecting a good set of predictive covariates without
considering statistical significance.
In this context, a first approach is to consider the multivariate linear model:
\begin{align*}
\*y = \*X\bm\beta^* + \bm\varepsilon \enspace ,
\end{align*}
where the target is denoted by $\*y \in \bbR^{n}$,
the design matrix by $\*X \in \bbR^{n \times p}$,
the parameter vector by $\bm\beta^* \in \bbR^{p}$ and
the random error vector by $\bm\varepsilon \in \bbR^{n}$.
The aim is to infer $\bm\beta^*$, with statistical guarantees on the
estimate, in particular regarding the support, \ie the set of covariates
with non-zero importance.

\paragraph{Statistical inference on individual parameters.}
In high-dimensional settings, standard statistical inference
methodology does not apply, but numerous methods have recently
been proposed to recover the non-zero parameters of $\bm\beta^*$
with statistical guarantees.
Many methods rely on resampling:
bootstrap procedures \citep{bach2008bolasso, chatterjee2011, liu2013asymptotic},
perturbation resampling-based procedures \citep{minnier2011perturbation},
stability selection procedures \citep{Meinshausen2010} and
randomized sample splitting \citep{wasserman2009, Meinshausen2008}.
All of these approaches suffer from limited power.
Contrarily to the screening/inference procedure, post-selection
inference procedures generally merge the screening and inference steps
into one and then use all the samples \citep{berk2013,
lockhart2014significance, lee2016, tibshirani2016exact}, resulting
in potentially more powerful tests than sample splitting.
Yet, these approaches do not scale well with large $p$.
Another family of methods rely on debiasing procedures:
the most prominent examples are corrected ridge \citep{buhlmann2013}
and desparsified Lasso \citep{Zhang_Zhang14,
vandeGeer_Buhlmann_Ritov_Dezeure14, Javanmard_Montanari14}
which is an active area of research
\citep{javanmard2018,bellec2019biasing, celentano2020lasso}.
Additionally, knockoff filters \citep{barber2015,candes2018}
consist in creating  noisy ``fake'' copies of the original variables, and
checking which original variables are selected prior to the fake ones.
Finally, a general framework for statistical inference in
sparse high-dimensional models has been proposed recently
\citep{ning2017general}.

\paragraph{Failure of existing statistical inference methods.}
In practice, in the $n \ll p$ setting we consider,
the previous methods are not well adapted as they are often
powerless or computationally intractable.
In particular, the number of predictive parameters (\ie the support size)
denoted $s(\bm\beta^*)$ can be greater than the number of
samples even in the sparse setting, where $s(\bm\beta^*) \ll p$.
There is an underlying identifiability problem: in general, one cannot
retrieve all predictive parameters, as highlighted \eg
in \citet{wainwright2009}.
Beyond the fact that statistical inference is impossible when $p \gg n$,
the problem is aggravated by the following three effects.
First, as outlined above, dense covariate sampling leads to high values
for $p$ and induces high correlation among covariates, further challenging
the conditions for recovery, as shown in \cite{wainwright2009}.
Second, when testing for several multiple hypothesis, the correction
cost is heavy \citep{dunn1961, westfall1993, Benjamini1995};
for example with Bonferroni correction \citep{dunn1961},
p-values are corrected by a factor $p$ when testing every covariate.
This make this type of inference methods powerless in our settings
(see \Cref{fig:estimated_weights} for instance).
Third, the above approaches are at least quadratic or cubic in the
support size, hence become prohibitive whenever both $p$ and
$n$ are large.

\paragraph{Combining clustering and inference.}
Nevertheless, in these settings, variables often reflect
some underlying spatial structure, such as smoothness.
For example, in medical imaging, an image has a $3$D structure and
a given voxel is highly correlated with neighboring voxels; in genomics,
there exist blocks of Single Nucleotide Polymorphisms (SNPs) that tend to
be jointly predictive or not.
Hence, $\bm\beta^*$ can in general be assumed to share the same
structure: among several highly correlated covariates, asserting that
only one is important to predict the target seems meaningless, if not
misleading.

A computationally attractive solution that alleviates high dimensionality
 is to group correlated neighboring covariates.
This step can be understood as a design compression: it produces a
closely related, yet reduced version of the original problem (see \eg
\citet{park2006, Varoquaux2012, hoyos2018recursive}).
Inference combined with a fixed clustering has been proposed by
\citet{Buhlmann2012} and can overcome the dimensionality issue, yet this
study does not provide procedures that derive cluster-wise confidence
intervals or p-values.
Moreover, in most cases groups (or clusters) are not pre-determined
nor easily identifiable from data, and their estimation simply
represents a local optimum among a huge, non-convex space of solutions.
It is thus problematic to base inference upon such an arbitrary
data representation.
Inspired by this dimension reduction approach,
we have proposed \citep{Chevalier2018a} the
ensemble of clustered desparsified Lasso (EnCluDL) procedure
that exhibits strong empirical performances \citep{chevalier2021decoding}
in terms of support recovery even when $p \gg n$.
EnCluDL is an ensembled clustered inference algorithm, \ie it
combines a spatially constrained clustering procedure
that reduces the problem dimension, an inference
procedure that performs statistical inference at the cluster level,
and an ensembling method that aggregates several cluster-level solutions.
Concerning the inference step, the desparsified Lasso
\citep{Zhang_Zhang14,vandeGeer_Buhlmann_Ritov_Dezeure14,Javanmard_Montanari14}
was preferred over other high-dimensional statistical inference procedures
based on the comparative study of \citet{Dezeure2015}
and on the research activity around it
\citep{dezeure2017,javanmard2018,bellec2019biasing,celentano2020lasso};
however, it is be possible to use another inference procedure that produces
a p-value family controlling the classical FWER.
By contrast, we did not consider the popular knockoff procedure
\citep{barber2015, candes2018}, that does not produce p-values
and does not control the family-wise error rate (FWER).
However, an extension of the knockoffs to FWER-type control was proposed
by \citet{janson2016}.
It does not control the standard FWER but another relaxed version of the FWER
called $k$-FWER.
As it is a relevant alternative to ensembled clustered inference algorithms,
we have included it in our empirical comparison
(see \Cref{sec:simulations_ecdl_proof}).
In \citet{nguyen2020}, a variant of the knockoffs is proposed
to control the FWER, but it does not handle large-p problems.
Another extension that produces p-value, called conditional randomization test,
has been presented in \citet{candes2018}, but its computational cost is prohibitive.
Additionally, \citet{meinshausen2015group} provides ``group-bound''
confidence intervals, corresponding to confidence intervals on the
$\ell_1$-norm of several parameters, without making further assumptions
on the design matrix.
However, this method is known to be conservative in practice
\citep{Mitra_Zhang16,javanmard2018}.
Finally, hierarchical testing \citep{Mandozzi:2016,
blanchard2005hierarchical,meinshausen2008hierarchical}
also leverages this clustering/inference combination but in a different way.
Their approach consists in performing significance tests along the
tree of a hierarchical clustering algorithm starting from the root
node and descending subsequently into children of rejected nodes.
This procedure has the drawback of being constrained by the clustering tree, which is often not available, thus replaced by some noisy estimate.

\paragraph{Contributions.}
Producing a cluster-wise inference is not completely satisfactory
as it relies on an arbitrary clustering choice.
Instead, we look for methods that derive covariate-wise statistics
enabling support identification with a spatially relaxed false detection
control.
In that regard, our first contribution is to present a generalization
of the FWER called $\delta$-FWER, that takes into account a spatial
tolerance of magnitude $\delta$ for the false discoveries.
Then, our main contribution is to prove that
ensembled clustered inference algorithms control
the $\delta$-FWER under reasonable assumptions for
a given tolerance parameter $\delta$.
Finally, we apply the ensembled clustered inference
scheme to the desparsified Lasso leading to the
EnCluDL algorithm and conduct an empirical study:
we show that EnCluDL exhibits a good statistical power
in comparison with alternative procedures and we verify that
it displays the expected $\delta$-FWER control.

\paragraph{Notation.}
Throughout the remainder of this article,
for any $p \in \bbN^*$, we write $ [p]$ for the set $\{1, \ldots, p\}$.
For a vector $\bm\beta$, $\bm\beta_{j}$ refers to its $j$-th coordinate.
For a matrix $\*X$, $\*X_{i,.}$ refers to the $i$-th row and
$\*X_{.,j}$ to the $j$-th column and $\*X_{i,j}$ refers to the
element in the $i$-th row and $j$-th column.

\section{Model and data assumptions}
\label{sec:context}

\subsection{Generative models of high-dimensional data: random fields}
\label{sec:design_matrix}
%
In the setting that we consider, we assume that the covariates come
with a natural representation in a discretized metric space,
generally the discretized $2$D or $3$D Euclidean space.
In such settings, discrete random fields are convenient
to model the random variables representing the covariates.
Indeed, denoting by $\*X = (\*X_{i,j})_{i \in [n],j \in [p]}$ the random
design matrix, where $n$ is the number of samples and $p$ the number of
covariates, the rows $(\*X_{i, .})_{i \in [n]}$ are sampled from
a random field defined on a discrete domain.

\subsection{Gaussian random design model and high dimensional settings}
\label{sec:linear_model}
%
We assume that the covariates are independent
and identically distributed and follow a centered
Gaussian distribution, \ie for all $i \in [n]$,
$\*X_{i,.} \sim \mathcal{N}(0_{p},\bm\Sigma)$ where $\bm\Sigma$
is the covariance matrix of the covariates.
Our aim is to derive confidence bounds or p-values on the coefficients
of the parameter vector denoted by $\bm\beta^*$, under the
Gaussian linear model:
\begin{align}\label{eq:noise_model_theory}
\*y = \*X\bm\beta^* + \bm\varepsilon \enspace ,
\end{align}
where
$\*y \in \bbR^{n}$ is the target, $\*X \in \bbR^{n \times p}$ is the
(random) design matrix, $\bm\beta^* \in \bbR^{p}$ is the vector or parameters,
and $\bm\varepsilon \sim \mathcal{N}(0,\sigma_{\varepsilon}^2 \*I_n)$ is the
noise vector with standard deviation $\sigma_{\varepsilon}>0$.
We make the assumption that $\bm\varepsilon$ is independent of $\*X$.
%
\subsection{Data structure}
\label{sec:data_structure}
%
Since the covariates have a natural representation in a metric space,
we assume that the spatial distances between covariates are known.
With a slight abuse of notation, the distance between covariates
$j$ and $k$ is denoted by $d(j,k)$ for $(j,k) \in [p] \times [p]$
and the correlation between covariates $j$ and $k$
is given by $\Cor(\*X_{.,j}, \*X_{.,k}) =
\bm\Sigma_{j, k} / \sqrt{\bm\Sigma_{j, j} \bm\Sigma_{k, k}}$.
We now introduce a key structural assumption:
two covariates at a spatial distance smaller
than $\delta$ are positively correlated.
\begin{ass}\label{ass:assumption_1}
The covariates verify the \emph{spatial homogeneity
assumption} with distance parameter $\delta > 0$ if, for all
$(j,k) \in [p] \times [p]$, $d(j, k) \leq \delta$
implies that $\bm\Sigma_{j, k} \geq 0$.
\end{ass}
Under model \Cref{eq:noise_model_theory}, each coordinate of the
parameter vector $\bm\beta^*$ links one covariate to the target.
Then, $\bm\beta^*$ has the same underlying organization as the covariates
and is also called weight map in these settings.
Defining its \emph{support} as
$S(\bm\beta^*) = \discsetin{j \in [p]:\bm\beta^*_j \neq 0}$
and its cardinal as $s(\bm\beta^*) = |S(\bm\beta^*)|$,
we assume that the true model is sparse, meaning that $\bm\beta^*$
has a small number of non-zero entries, \ie $s(\bm\beta^*) \ll p$.
The complementary of $S(\bm\beta^*)$ in $[p]$ is
called the \emph{null region} and is denoted by $N(\bm\beta^*)$,
\ie $N(\bm\beta^*)  = \discsetin{j \in [p]:\bm\beta^*_j = 0}$.
Additionally to the sparse assumption, we assume that $\bm\beta^*$
is (spatially) smooth.
To reflect sparsity and smoothness, we introduce another key assumption:
weights associated with close enough covariates share the same sign,
zero being both positive and negative.
\begin{ass}\label{ass:assumption_2}
The weight vector $\bm\beta^*$ verifies the \emph{sparse-smooth assumption}
with distance parameter $\delta > 0$ if, for all
$(j,k) \in [p] \times [p]$, $d(j, k) \leq \delta$
implies that $\sign(\bm\beta_j^*) = \sign(\bm\beta_k^*)$.
\end{ass}
Equivalently, the sparse-smooth assumption with parameter $\delta$
holds if the distance between the two closest weights of opposite
sign is larger than $\delta$.
In \Cref{fig:full_fig}-(a), we give an example of a weight map verifying
the sparse-smooth assumption with $\delta = 2$.
%
\section{Statistical control with spatial tolerance}
\label{sec:definitions}
%
Under the spatial assumption we have discussed, discoveries that are closer than $\delta$ from the true support are not considered as false discoveries:  inference at a resolution finer than $\delta$ might be unrealistic.
This means that $\delta$ can be interpreted as a tolerance parameter on the (spatial) support we aim at recovering.
Then, we introduce a new metric closely related to the
FWER that takes into account spatial tolerance and we
call it $\delta$-family wise error rate ($\delta$-FWER).
A similar extension of the false discovery rate (FDR) has been
introduced by \citet{cheng2020multiple, nguyen2019, gimenez2019}, but,
to the best of our knowledge, this has not been considered yet for the FWER.
In the following, we consider a general estimator
$\hat{\bm\beta}$ that comes with p-values, testing the nullity
of the corresponding parameters, denoted by
$\hat{p} = (\hat{p}_j)_{j \in [p]}$.
Also, we denote by $S(\hat{\bm\beta}) \subset [p]$ a general
estimate of the support $S(\bm\beta^*)$ derived from
the estimator $\hat{\bm\beta}$.
\begin{df}[$\delta$-null hypothesis]\label{df:delta_null_hypothesis}
For all $j \in [p]$, the $\delta$-null hypothesis for the $j$-th covariates,
$H^{\delta}_0(j)$, states that all other covariates at distance
less than $\delta$ have a zero weight in the true model \Cref{eq:noise_model_theory};
the alternative hypothesis is denoted $H^{\delta}_1(j)$:
\begin{align*}
\begin{split}
& H^{\delta}_0(j) : \text{``for all}~ k \in [p] ~ \text{such that}~ d(j, k) \leq \delta,
~ \bm\beta^*_k = 0\text{''} \enspace , \\
& H^{\delta}_1(j) : \text{``there exists}~ k \in [p] ~ \text{such that}~ d(j, k) \leq \delta ~
\text{and }\, \bm\beta^*_k \neq 0\text{''} \enspace . \\
\end{split}
\end{align*}
\end{df}
Thus, we say that a $\delta$-type 1 error is made if
a null covariate $j \in [p]$ is selected, \ie $j \in S(\hat{\bm\beta})$,
while $H^{\delta}_0(j)$ holds true.
Taking $\delta=0$ recovers the usual
null-hypothesis $H_0(j) :  \text{``}\bm\beta^*_j = 0$'' and
usual type 1 error.

\begin{df}[Control of the $\delta$-type 1 error] \label{df:control_delta_type_one_error}
The p-value related to the $j$-th covariate denoted by $\hat{p}_j$
controls the $\delta$-type 1 error if, under $H^{\delta}_0(j)$,
for all $\alpha \in (0,1)$, we have:
\begin{align*}
\bbP(\hat{p}_j \leq \alpha) \leq \alpha \enspace ,
\end{align*}
where $\bbP$ is the probability distribution with respect to
the random dataset of observations $(\*X, \*y)$.
\end{df}
\begin{df}[$\delta$-null region]\label{df:delta_null_region}
The set of indexes of covariates verifying the $\delta$-null
hypothesis is called the $\delta$-null region and is denoted by
$N^{\delta}(\bm\beta^*)$ (or simply $N^{\delta}$):
\begin{align*}
N^{\delta}(\bm\beta^*) = \left\{j \in [p] : \text{for all}~ k \in [p],~
d(j,k) \leq \delta ~ \text{implies that}~ \bm\beta^*_k = 0 \right\} \enspace .
\end{align*}
\end{df}
When $\delta = 0$ the $\delta$-null region is simply
the null region : $N^{0}(\bm\beta^*)=N(\bm\beta^*)$.
We also point out the nested property of $\delta$-null regions with respect to $\delta$:
for $0 \leq \delta_1 \leq \delta_2$ we have $N^{\delta_2}(\bm\beta^*)
\subseteq N^{\delta_1}(\bm\beta^*) \subseteq N(\bm\beta^*)$ (see \Cref{fig:full_fig}-(d) for an example of $\delta$-null region).

\begin{df}[Rejection region]\label{df:rejection_region}
Given a family of p-values $\hat{p} = (\hat{p}_j)_{j \in [p]}$ and
a threshold $\alpha \in (0,1)$, the rejection region, $R_{\alpha}(\hat{p})$, is the set of indexes having a p-value lower than $\alpha$:
\begin{align*}
R_{\alpha}(\hat{p}) = \left\{j \in [p] : \hat{p}_j \leq \alpha \right\}
\enspace .
\end{align*}
\end{df}
\begin{df}[$\delta$-type 1 error region]\label{df:delta_type_one_error_region}
Given a family of p-values $\hat{p} = (\hat{p}_j)_{j \in [p]}$ and a
threshold $\alpha \in (0,1)$, the $\delta$-type 1 error region at level
$\alpha$ is $\mathscr{E}_{\alpha}^{\delta}$, the set of indexes belonging both to the $\delta$-null region and to the rejection region at level $\alpha$.
We also refer to this region as the erroneous rejection region at
level $\alpha$ with tolerance $\delta$:
\begin{align*}
\mathscr{E}^{\delta}_{\alpha}(\hat{p}) = N^{\delta} \cap R_{\alpha}(\hat{p})
\enspace .
\end{align*}
\end{df}
When $\delta = 0$ the $\delta$-type 1 error region recovers the type
1 error region which is denoted by $\mathscr{E}_{\alpha}(\hat{p})$.
Again, one can verify a nested property: for
$0 \leq \delta_1 \leq \delta_2$ we have
$\mathscr{E}_{\alpha}^{\delta_2}(\hat{p}) \subseteq
\mathscr{E}_{\alpha}^{\delta_1}(\hat{p}) \subseteq
\mathscr{E}_{\alpha}(\hat{p})$.
\begin{df}[$\delta$-family wise error rate]\label{df:delta_fwer}
Given a family of p-values $\hat{p} = (\hat{p}_j)_{j \in [p]}$ and a
threshold $\alpha \in (0,1)$, the $\delta$-FWER at level $\alpha$
with respect to the family $\hat{p}$, denoted $\delta\mbox{-FWER}_{\alpha}(\hat{p})$,
is the probability that the $\delta$-type 1 error region at level $\alpha$
is not empty:
\begin{align*}
\delta\mbox{-FWER}_{\alpha}(\hat{p})
= \bbP(|\mathscr{E}^{\delta}_{\alpha}(\hat{p})| \geq 1)
= \bbP(\min_{j \in N^{\delta}}\hat{p}_j \leq \alpha )
\enspace .
\end{align*}
\end{df}
\begin{df}[$\delta$-FWER control]\label{df:control_delta_fwer}
We say that the family of p-values $\hat{p} = (\hat{p}_j)_{j \in [p]}$
controls the $\delta$-FWER if, for all $\alpha \in (0,1)$:
\begin{align*}
\delta\mbox{-FWER}_{\alpha}(\hat{p}) \leq \alpha \enspace .
\end{align*}
\end{df}
When $\delta = 0$ the $\delta$-FWER is the usual FWER.
Additionally, for $0 \leq \delta_1 \leq \delta_2$, one can verify that
$\delta_2\mbox{-FWER}_{\alpha}(\hat{p}) \leq
\delta_1\mbox{-FWER}_{\alpha}(\hat{p}) \leq
\mbox{FWER}_{\alpha}(\hat{p})$.
Thus, $\delta$-FWER control is a weaker property than usual FWER control.
%
\section{$\delta$-FWER control with clustered inference algorithms}
\label{sec:cdl_properties}

\subsection{Clustered inference algorithms}
\label{sec:encludl_presentation}
%
A clustered inference algorithm consists in partitioning the
covariates into groups (or clusters) before applying a
statistical inference procedure.
In \Cref{alg:CluDL_new}, we describe a standard clustered inference
algorithm that produces a (corrected) p-value family on the
parameters of the model \Cref{eq:noise_model_theory}.
In this algorithm, in addition to the observations ($\*X, \*y$),
we take as input the transformation matrix $\*A \in \bbR^{p \times C}$
which maps and averages covariates into $C$ clusters.
%
%
The \texttt{statistical\_inference} function corresponds to a given
statistical inference procedure that takes as inputs the clustered data $\*Z$
and the target $\*y$ and produces valid p-values for every cluster.
If $C < n$, least squares are suitable, otherwise, procedures such
as multi-sample split \citep{wasserman2009, Meinshausen2008},
corrected ridge \citep{buhlmann2013} or desparsified Lasso
\citep{Zhang_Zhang14,vandeGeer_Buhlmann_Ritov_Dezeure14,Javanmard_Montanari14}
might be relevant whenever their assumptions are verified.
Then, the computed p-values are corrected for multiple testing
by multiplying by a factor $C$.
Finally, covariate-wise p-values are inherited from the corresponding cluster-wise p-values.

{\fontsize{4}{4}\selectfont
\begin{algorithm}\label{alg:CluDL_new}
\SetKwInOut{Input}{input}
\SetKwInOut{Output}{output}
\SetKwInOut{Parameter}{param}
\caption{Clustered inference}

\Input{$\*X \in \bbR^{n \times p}, \*y \in \bbR^{n}, \*A \in \bbR^{p \times C}$}

\vspace {1mm}

$\*Z = \*X \*A$
\tcp*{compressed design matrix}
$\hat{p}^{\mathcal{G}} = \texttt{statistical\_inference}(\*Z, \*y)$
\tcp*{uncorrected cluster-wise p-values}
$\hat{q}^{\mathcal{G}} = C \times \hat{p}^{\mathcal{G}}$
\tcp*{corrected cluster-wise p-values}
\vspace {1mm}
\For{$j = 1, \dots, p $}
{
  $\hat{q}_j = \hat{q}^{\mathcal{G}}_c$ if $j$ in cluster $c$
  \tcp*{corrected covariate-wise p-values}
}
\vspace {1mm}
\Return $\hat{q} = (\hat{q}_j)_{j \in [p]}$
\tcp*{family of corrected covariate-wise p-values}

\end{algorithm}
}

{\fontsize{4}{4}\selectfont
\begin{algorithm}\label{alg:EnCluDL_new}
\SetKwInOut{Input}{input}
\SetKwInOut{Output}{output}
\SetKwInOut{Parameter}{param}
\caption{Ensembled clustered inference}

\Input{$\*X \in \bbR^{n \times p}, \*y \in \bbR^{n}$}

\vspace {1mm}

\Parameter{$C, B$}

\vspace {1mm}

\For{$b = 1, \dots, B$}
{\vspace {1mm}
  $ \*X^{(b)} = \texttt{sampling}(\*X)$
  \tcp*{sampling rows of $\*X$}
  $ \*A^{(b)} = \texttt{clustering}(q, \*X^{(b)})$
  \tcp*{transformation matrix}
  $ \hat{q}^{(b)} = \texttt{clustered\_inference}(\*X, \*y, \*A^{(b)})$
  \tcp*{families of corr.~covariate-wise p-val.}
}
\vspace {1mm}
\For{$j = 1, \dots, p $}
{
  $\hat{q}_j = \texttt{ensembling}(\{\hat{q}_j^{(b)}, b \in [B]\})$
  \tcp*{aggregated corrected covariate-wise p-values}
}
\vspace {1mm}
\Return $\hat{q} = (\hat{q}_j)_{j \in [p]}$
\tcp*{family of aggregated corrected covariate-wise p-values}
\end{algorithm}
}

Ensembled clustered inference algorithms correspond to
the ensembling of several clustered inference solutions
for different choice of clusterings using the p-value
aggregation proposed by \citet{Meinshausen2008}.
In \Cref{alg:EnCluDL_new}, we give a standard ensembled clustered
inference algorithm that produces a (corrected) p-value family on the
parameters of the model \Cref{eq:noise_model_theory}.
In this algorithm, the \texttt{sampling} function  corresponds to
a subsampling of the data, \ie a subsampling of the rows of $\*X$.
The \texttt{clustering} function derives a choice of clustering in $C$
clusters, it produces a transformation matrix $\*A^{(b)} \in \bbR^{p \times C}$
that should vary for each bootstrap $b \in [B]$ since the subsampled data
$\*X^{(b)}$ varies.
Once the clustering inference steps are completed,
the \texttt{ensembling} function aggregates the $B$
(corrected) p-value families into a single one.

\Cref{fig:organization} can help the reader to better understand
the organization of the next sections, aiming eventually at establishing
the $\delta$-FWER control property of the clustered inference and
ensembled clustered inference algorithms.

\begin{figure}[!ht]
    \centering
    \includegraphics[width=1.0\linewidth]{organisation.pdf}
    \caption{Organization of \Cref{sec:cdl_properties}.}
    \label{fig:organization}
\end{figure}

\subsection{Compressed representation}
\label{sec:compressed_representation}
%
The motivation for using groups of covariates that are spatially
concentrated is to reduce the dimension while preserving
large-scale data structure.
The number of groups is denoted by $C < p$ and,
for $r \in [q]$, we denote by $G_r$ the $r$-th group.
The collection of all the groups is denoted by
$\mathcal{G} = \{G_1, G_2, \ldots, G_C\}$ and
forms a partition of $[p]$.
Every group representative variable is defined by the average of the
covariates it contains.
Then, denoting by $\*Z \in \bbR^{n \times C}$ the compressed random
design matrix that contains the group representative variables in
columns and, without loss of generality, assuming a suitable
ordering of the columns of $\*X$,
dimension reduction can be written:
\begin{align}\label{eq:clustering_theory}
\*Z = \*X\*A \enspace ,
\end{align}
where $\*A \in \bbR^{p \times q}$ is the transformation matrix
defined by:
\begin{align*}
\*A = \left[
\begin{matrix}
\alpha_1 \horzbar \alpha_1 & 0 \horzbar 0 &\ldots & 0 \horzbar 0\\
0 \horzbar 0  & \alpha_2  \horzbar \alpha_2 & \ldots & 0 \horzbar 0\\
 \vdots & \vdots & \ddots & \vdots  \\
0 \horzbar 0  & 0 \horzbar 0 & \ldots & \alpha_C \horzbar \alpha_C\\
\end{matrix}
\right] \enspace ,
\end{align*}
where $\alpha_c = {1}/{|G_c|}$ for all $c \in [C]$.
Consequently, the distribution of the $i$-th row of $\*Z$ is given by
$\*Z_{i,.} \sim \mathcal{N}_{q}(0, \bm\Upsilon)$, where
$\bm\Upsilon = \*A^\top \bm\Sigma \*A$.
The correlation between the groups $r \in [q]$ and $l \in [q]$ is given by
$\Cor(\*Z_{.,r},\*Z_{.,l}) = \bm\Upsilon_{r, l} / \sqrt{\bm\Upsilon_{r, r}\bm\Upsilon_{l, l}}$.
As mentioned in \citet{Buhlmann2012}, because of the Gaussian assumption in
\Cref{eq:noise_model_theory}, we have the following compressed representation:
\begin{align}\label{eq:noise_model_compressed}
\*y = \*Z\bm\theta^* + \bm\eta \enspace ,
\end{align}
where $\bm\theta^* \in \bbR^{q}$,
$\bm\eta \sim \mathcal{N}(0,\sigma_{\eta}^2 \*I_n)$,
$\sigma_{\eta} \geq \sigma_{\varepsilon}>0$ and
$\bm\eta$ is independent of $\*Z$.
\begin{rk}\label{improved_conditioning}
  Dimension reduction is not the unique desirable effect
  of clustering with regards to statistical inference.
  Indeed, this clustering-based design compression
  also generally improves the conditioning of the problem.
  Assumptions needed for valid statistical inference are thus
  more likely to be met.
  For more details about this conditioning enhancement,
  the reader may refer to \citet{Buhlmann2012}.
\end{rk}
%
\subsection{Properties of the compressed model weights}
\label{sec:compressed_representation_property}
%
\begin{figure}[!ht]
    \centering
    \includegraphics[width=1.0\linewidth]{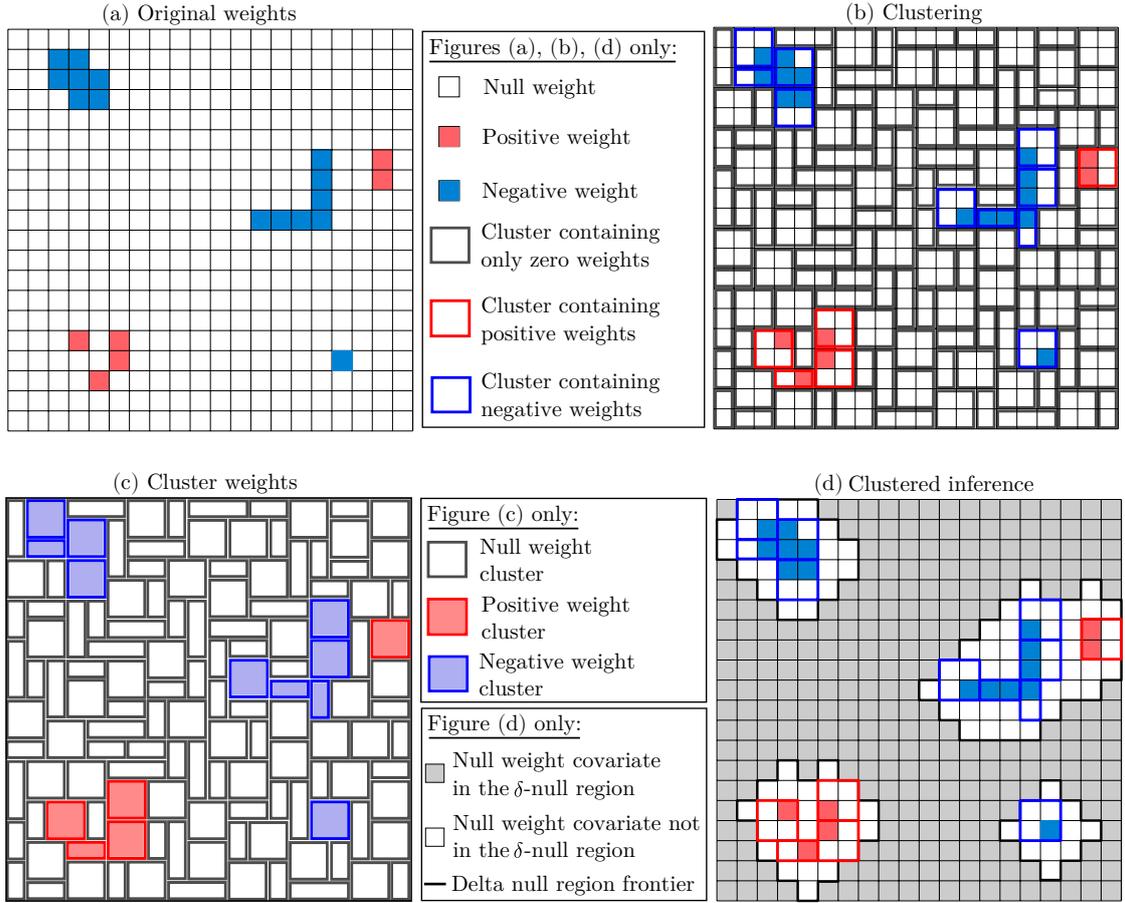}
    \caption{Clustered inference mechanism on 2D-spatially structured data.
    \underline{Item a:}
    Example of weight map with a 2D-structure.
    Voxels represent covariates, with blue (resp. red) corresponding to negative (resp. positive) weights; others are null weights.
    \underline{Item b:}
    Arbitrary choice of spatially constrained clustering
    with a diameter of $\delta=2$ units for the $\ell_1$-distance.
    Rectangles delimited by black lines represent clusters that
    contain only zero-weight covariates. Blue (resp. red) rectangles
    refer to clusters that contain negative-weight (resp. positive) covariates.
    \underline{Item c:}
    Compressed model weights: under the assumptions of
    \Cref{prop:weights_compressed}, the cluster weights
    share the same signs as the covariate weights they contain.
    Blue (resp. red) rectangles correspond to negative-weight
    (resp. positive-weights) clusters.
    \underline{Item d:}
    The grey area corresponds to the $\delta$-null region ($\delta = 2$).
    Under the same assumptions, the non-zero weight groups have no
    intersection with the $\delta$-null region.}
    \label{fig:full_fig}
\end{figure}
We now give a property of the weights of the compressed problem
which is a consequence of \citet[Proposition 4.3]{Buhlmann2012}.
\begin{prop}\label{prop:weights_compressed}
Considering the Gaussian linear model in \Cref{eq:noise_model_theory}
and assuming:
\begin{align*}
\begin{split}
& (i) ~ \text{for all}~ c \in [C], ~ \text{for all}~ (j, k) \in (G_c)^2,
~ \bm\Sigma_{j, k} \geq 0 \enspace , \\
& (ii) ~ \text{for all}~ c \in [C], ~ \text{for all}~ c^\prime \in [C] \setminus \{ c \},
~ \bm\Upsilon_{c, c^\prime} = 0 \enspace , \\
& (iii) ~ \text{for all}~ c \in [C],
\left(\bm\beta^*_j \geq 0 ~ \text{for all}~ j \in G_c \right) \orr
\left(\bm\beta^*_j \leq 0 ~ \text{for all}~ j \in G_c \right) \enspace , \\
\end{split}
\end{align*}
then, in the compressed representation \Cref{eq:noise_model_compressed},
for $c \in [C]$, $\bm\theta_c^* \neq 0$ if and only if there exists
$j \in G_c$ such that $\bm\beta_j^* \neq 0$.
If such an index $j$ exists then
$\sign(\bm\theta_r^*) = \sign(\bm\beta_j^*)$.
\end{prop}
\begin{proof}
See \Cref{sec:proof_weights_compressed}.
\end{proof}
Assumption $(i)$ states that the covariates in a group are all
positively correlated.
Let us define the group diameter (or cluster diameter)
of $G_c$ by the distance
that separates its two most distant covariates, \ie
$\Diam(G_c) = \max \{d(j, k) : (j, k) \in (G_c)^2\}$
and the clustering diameter of $\mathcal{G}$
by the largest group diameter, \ie
$\Diam(\mathcal{G}) = \max \{\Diam(G_c) : c \in [C] \}$.
In \Cref{fig:full_fig}-(b), we propose a clustering of
the initial weight map in \Cref{fig:full_fig}-(a) for which the
clustering diameter is equal to $2$ for the $\ell_1$ distance.
Assumption $(i)$ notably holds when $\Diam(\mathcal{G}) \leq \delta$ under the spatial homogeneity
assumption (\Cref{ass:assumption_1}) with parameter $\delta$.
Assumption $(ii)$ assumes independence of the groups.
A sufficient condition is when the covariates covariance
matrix $\bm\Sigma$ is block diagonal, with blocks coinciding
with the group structure; \ie assumption $(ii)$ holds when covariates
of different groups are independent.
In practice, this assumption may be unmet, and we relax it in
\Cref{sec:cdl_properties_corr}.
Assumption $(iii)$ states that all the weights in a group share
the same sign.
This is notably the case when the clustering diameter is smaller
than $\delta$ and the weight map satisfies the sparse-smooth
assumption (\Cref{ass:assumption_2}) with parameter $\delta$.
For instance, a clustering-based compressed representation
of the weight map in \Cref{fig:full_fig}-(a) is given in
\Cref{fig:full_fig}-(c).
%
\subsection{Statistical inference on the compressed model}
\label{sec:stat_inf}
%
To perform the statistical inference on the compressed problem
\Cref{eq:noise_model_compressed}, we could consider any statistical
inference procedure that produces cluster-wise p-values
$\hat{p}^{\mathcal{G}} = (\hat{p}^{\mathcal{G}}_c)_{c \in [C]}$,
given a choice of clustering $\mathcal{G}$, that control the type 1 error.
More precisely, for any $c \in [C]$, under $H_{0}(G_c)$,
\ie the null hypothesis which states that $\theta_c^*$ is
equal to zero in the compressed model, we assume that
the p-value associated with the $c$-th cluster verifies:
\begin{align}\label{eq:p_value_type_1_error_control}
\begin{split}
  \bbP(\hat{p}^{\mathcal{G}}_c \leq \alpha)
  \leq \alpha \enspace .
\end{split}
\end{align}
To correct for multiple comparisons, we consider Bonferroni correction
\citep{dunn1961} which is a conservative procedure but has the advantage
of being valid without any additional assumptions.
Furthermore, here the correction factor is only equal to the number of
groups, not the number of covariates.
Then, the family of corrected cluster-wise p-values
$\hat{q}^{\mathcal{G}} = (\hat{q}^{\mathcal{G}}_c)_{c \in [C]}$
is defined by:
\begin{align}\label{eq:corrected_p_value}
\hat{q}^{\mathcal{G}}_c = \min\{1, C \times {\hat{p}^{\mathcal{G}}_c}\}
\enspace .
\end{align}
Let us denote by $N_{\mathcal{G}}(\bm\theta^*)$ (or simply
$N_{\mathcal{G}}$) the null region in the compressed problem
for a given choice of clustering $\mathcal{G}$, \ie
$N_{\mathcal{G}}(\bm\theta^*) = \left\{c \in [C] : \bm\theta^*_c = 0 \right\}$.
Then, for all $\alpha \in (0, 1)$:
\begin{align}\label{eq:control_fwer_2}
\mbox{FWER}_{\alpha}(\hat{q}^{\mathcal{G}}) =
\bbP(\min_{c \in N_{\mathcal{G}}}\hat{q}^{\mathcal{G}}_c \leq \alpha)
\leq \alpha \enspace .
\end{align}
This means that the cluster-wise p-value family $\hat{q}^{\mathcal{G}}$
controls FWER.

\subsection{De-grouping}
\label{sec:de-grouping}
%
Given the families of cluster-wise p-values
$\hat{p}^{\mathcal{G}}$ and corrected p-values $\hat{q}^{\mathcal{G}}$
as defined in \Cref{eq:p_value} and \Cref{eq:corrected_p_value},
our next aim is to derive families of p-values and corrected p-values
related to the covariates of the original problem.
To construct these families, we simply set the (corrected) p-value of
the $j$-th covariate to be equal to the (corrected) p-value of its
corresponding group:
\begin{align}\label{eq:de_grouping_p_values}
\begin{split}
& \text{for all}~ j \in [p], \quad \hat{p}_j =
\sum_{c \in [C]} \mathds{1}_{\{j \in G_c\}} ~ \hat{p}^{\mathcal{G}}_{c}
\enspace , \\
& \text{for all}~ j \in [p], \quad \hat{q}_j =
\sum_{c \in [C]} \mathds{1}_{\{j \in G_c\}} ~ \hat{q}^{\mathcal{G}}_{c}
\enspace .\\
\end{split}
\end{align}
\begin{prop}\label{prop:p_values_control}
Under the assumptions of \Cref{prop:weights_compressed}
and assuming that the clustering diameter is smaller
than $\delta$, then:

\medskip

\noindent (i) elements of the family $\hat{p}$ defined in
\Cref{eq:de_grouping_p_values} control the $\delta$-type 1 error:
\begin{align*}
\text{for all}~ j \in N^{\delta}, ~ \text{for all}~ \alpha \in (0, 1),
~ \bbP(\hat{p}_{j} \leq \alpha) \leq \alpha  \enspace ,
\end{align*}
(ii) the family $\hat{q}$ defined in \Cref{eq:de_grouping_p_values}
controls the $\delta$-FWER:
\begin{align*}
\text{for all}~ \alpha \in (0, 1),
~ \bbP (\min_{j \in N^{\delta}} (\hat{q}_j) \leq \alpha) \leq \alpha
\enspace .
\end{align*}
\end{prop}
\begin{proof}
See \Cref{sec:proof_degrouping}.
\end{proof}
The previous de-grouping properties can be seen in
\Cref{fig:full_fig}-(d).
Roughly, since all the clusters that intersect
the $\delta$-null region have low p-value
with low probability, one can conclude that all the
covariates of the $\delta$-null region also have
low p-value with low probability.

\subsection{Ensembling}
\label{sec:ensembling_theory}
%
Let us consider $B$ families of corrected p-values
that control the $\delta$-FWER.
For any $b \in [B]$, we denote by $\hat{q}^{(b)}$ the $b$-th family
of corrected p-values.
Then, we show that the ensembling method proposed in
\citet{Meinshausen2008} yields a family that also enforces
$\delta$-FWER control.
\begin{prop}\label{prop:p_values_aggregation}
Assume that, for $b \in [B]$,
the p-value families $\hat{q}^{(b)}$
control the $\delta$-FWER.
Then, for any $\gamma \in (0, 1)$, the ensembled p-value family
$\tilde{q} (\gamma)$ defined by:
\begin{align}\label{eq:p_values_aggregation}
  \text{for all}~ j \in [p], ~ \tilde{q}_j (\gamma) =
  \min \left\{ 1, \gamma\mbox{-quantile}
  \left( \left\{ \frac{\hat{q}^{(b)}_j}{\gamma} :
  b \in [B] \right\} \right) \right\} \enspace ,
\end{align}
controls the $\delta$-FWER.
\end{prop}
\begin{proof}
See \Cref{sec:proof_aggregation}.
\end{proof}
%
\subsection{$\delta$-FWER control}
\label{sec:encludl_property_theory}
%
We can now state our main result: the clustered inference and
ensembled clustered inference algorithms control the $\delta$-FWER.
\begin{thm}\label{prop:ecdl_properties_final}
Assume the model given in \Cref{eq:noise_model_theory}
and that the data structure assumptions,
\Cref{ass:assumption_1} and \Cref{ass:assumption_2},
are satisfied for a distance parameter larger than $\delta$.
Assume that all the clusterings considered have
a diameter smaller than $\delta$.
Assume that the uncorrelated cluster assumption,
\ie assumption $(ii)$ of \Cref{prop:weights_compressed},
is verified for each clustering and further assume that
the statistical inference performed on the compressed model
\Cref{eq:noise_model_compressed} is valid,
\ie \Cref{eq:p_value_type_1_error_control} holds.
Then, the p-value family obtained from the clustered inference
algorithm controls the $\delta$-FWER.
Additionally, the p-value family derived by the ensembled
clustered inference algorithm controls the $\delta$-FWER.
\end{thm}
\begin{proof}
See \Cref{sec:proof_ecdl_properties_final}.
\end{proof}
\begin{rk}
  When the type 1 error control offered by the
  statistical inference procedure is only asymptotic,
  the result stated by \Cref{prop:ecdl_properties_final}
  remains true asymptotically.
  This is notably the case when using desparsified Lasso:
  under the assumptions of \Cref{prop:ecdl_properties_final} and
  the assumptions specific to desparsified Lasso (cf. \Cref{sec:stat_inf_dl}),
  ensemble of clustered desparsified Lasso (EnCluDL) controls
  the $\delta$-FWER asymptotically.
\end{rk}
\section{Numerical Simulations}
\label{sec:simulations_ecdl_proof}
%
\subsection{CluDL and EnCluDL}
\label{sec:CluDL_EnCluDL}

For testing the (ensembled) clustered inference algorithms,
we have decided to make the inference step using the desparsified Lasso
\citep{Zhang_Zhang14,vandeGeer_Buhlmann_Ritov_Dezeure14,Javanmard_Montanari14}
leading to the clustered desparsified Lasso (CluDL) and the ensemble
of clustered desparsified Lasso (EnCluDL) algorithms that were
first presented in \citet{Chevalier2018a}.

In \Cref{sec:stat_inf_dl}, we detail the assumptions and refinements
that occur when choosing the desparsified Lasso to perform the
statistical inference step.
A notable difference is the fact that all the results becomes asymptotic.
In \Cref{sec:algorithm_ECDL_journal}, we present a diagram illustrating
the mechanism of EnCluDL and analyse its numerical complexity.

\subsection{2D Simulation}
\label{sec:intro_2D_simulation}
%
We run a series of simulations on 2D data in order to give empirical
evidence of the theoretical properties of CluDL and EnCluDL and
compare their recovery properties with two other procedures.
For an easier visualization of the results, we consider one central scenario,
whose parameters are written in \textbf{bold} in the following of this section,
with several variations, changing only one parameter at a time.

In all these simulations, the feature space considered is a $2$D
square with edge length $H = 40$ leading to $p = H^{2} = \numprint{1600}$
covariates, with a sample size $n \in \{ 50, \mathbf{100}, 200, 400 \}$.
To construct $\bm\beta^*$, we define a $2$D weight map
$\tilde{\bm\beta}^*$ with four active regions (as illustrated
in \Cref{fig:estimated_weights}) and then flatten
$\tilde{\bm\beta}^*$ to a vector ${\bm\beta}^*$ of size $p$.
Each active region is a square of width $h \in \{ 2, \mathbf{4}, 6, 8 \}$,
leading to a size of support of $1\%$, $\mathbf{4\%}$, $9\%$ or $16\%$.
To construct the design matrix, we first build a $2$D data matrix
$\tilde{\*X}$ by drawing $p$ random normal vectors of size $n$ that
are spatially smoothed with a $2$D Gaussian filter to create
a correlation structure related to the covariates' spatial organization.
The same flattening process as before is used to get the design matrix
$\*X \in \bbR^{n \times p}$.
The intensity of the spatial smoothing is adjusted to achieve
a correlation between two adjacent covariates (local correlation) of
$\rho \in \{ 0.5, \mathbf{0.75}, 0.9, 0.95 \}$.
We also set the noise standard deviation
$\sigma_{\varepsilon} \in \{ 1, \mathbf{2}, 3, 4 \}$, which corresponds to a
signal to noise ratio (SNR) $\mbox{SNR}_{y} \in \{ 6.5, \mathbf{3.5}, 2.2, 1.5\}$,
where the SNR is defined by
$\SNR_{y} = \normin{\*X\bm\beta^*}_2 / \normin{\bm\varepsilon}_2$.
For each scenario, we run $100$ simulations
to derive meaningful statistics.
A Python implementation of the simulations and
procedures presented in this paper
is available on \texttt{https://github.com/ja-che/hidimstat}.
Regarding the clustering step in CluDL and EnCluDL, we used a spatially
constrained agglomerative clustering algorithm with Ward criterion.
This algorithm is popular in many applications \citep{Varoquaux2012,
dehman2015performance}, as it tends to create compact, balanced clusters.
Since the optimal number of clusters $C$ is unknown a
priori, we have tested several values $C \in [100 ; 400]$.
A smaller $C$ generally improves recovery,
but entails a higher spatial tolerance.
Following theoretical considerations, we compute the largest cluster diameter
for every value of $C$ and set $\delta$ to this value.
We obtained the couples
$(C,\delta) \in \{ (100,8), (200,6), (300,5), (400,4) \}$.
The tolerance region is represented in \Cref{fig:estimated_weights}
for $\delta = 6$.
Concerning EnCluDL, we took a number of bootstraps $B$ equal to $25$
as we observed that it was sufficient to benefit from most of the effect
of clustering randomization.
%
\subsection{Alternative methods}
\label{sec:alternative_methods}
%
We compare the recovering properties of CluDL and EnCluDL with two
other procedures: desparsified Lasso and knockoffs.
Contrarily to CluDL and EnCluDL, none of these includes a compression
step.
The version of the desparsified Lasso we have tested is the one presented in
\citet{vandeGeer_Buhlmann_Ritov_Dezeure14}, that outputs p-values.
Using Bonferroni correction it controls the classical FWER
at any desired rate.
The original version of knockoffs \citep{barber2015, candes2018}
only controls the false discovery rate (FDR) which is a weaker
control than the classical FWER.
Yet \citet{janson2016} modifies the covariate selection process
leading to a procedure that controls the $k$-FWER, \ie the probability of
making at least $k$ false discoveries.
We tested this last extension of knockoffs.
Depending on the nominal rate at which we want to control the $k$-FWER,
the choice of $k$ is not arbitrary.
More precisely, if we want a $k$-FWER control at $10 \%$,
we need to tolerate $k = 4$ at least, otherwise the estimated
support would always be empty.

Since $k$-FWER and $\delta$-FWER controls are both weaker than the usual
FWER control whenever $k > 1$ and $\delta > 0$, one can expect
desparsified Lasso to be less powerful than knockoffs, CluDL
and EnCluDL.
Besides, there is no relation between $k$-FWER
and $\delta$-FWER controls when $k > 1$ and $\delta > 0$,
hence it is not possible to establish which one
is less prohibitive for support recovery.
However, when data are spatially structured, $\delta$-FWER control might be
more relevant since it controls the very undesirable
far-from-support false discoveries.

\subsection{Results}
\label{sec:results_theory}
%
\begin{figure}[!ht]
      \centering
      \begin{minipage}{0.193\linewidth}
        \centering\includegraphics[width=\linewidth]
        {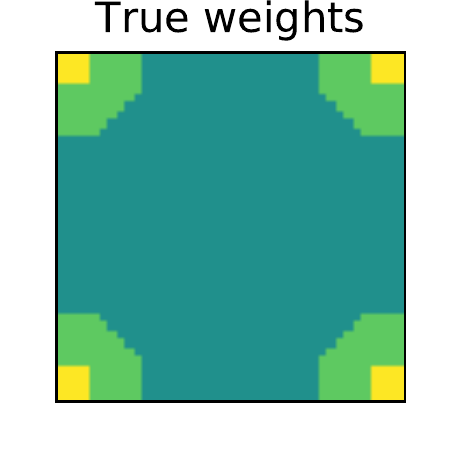}
      \end{minipage}
      \begin{minipage}{0.193\linewidth}
        \centering\includegraphics[width=\linewidth]
        {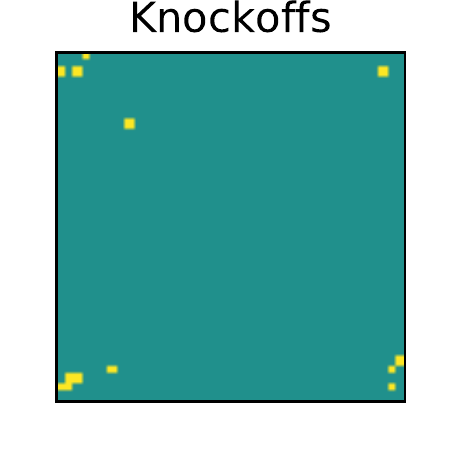}
      \end{minipage}
      \begin{minipage}{0.193\linewidth}
        \centering\includegraphics[width=\linewidth]
        {new_2D_dl.pdf}
      \end{minipage}
      \begin{minipage}{0.193\linewidth}
        \centering\includegraphics[width=\linewidth]
        {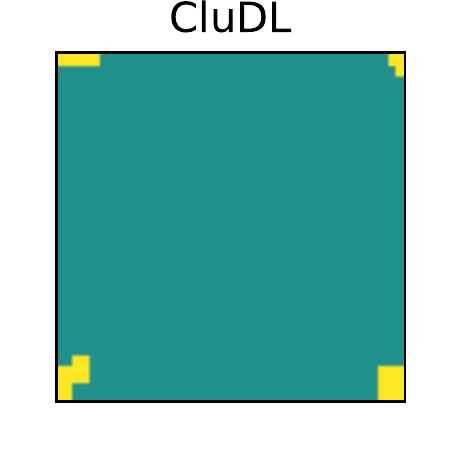}
      \end{minipage}
      \begin{minipage}{0.193\linewidth}
        \centering\includegraphics[width=\linewidth]
        {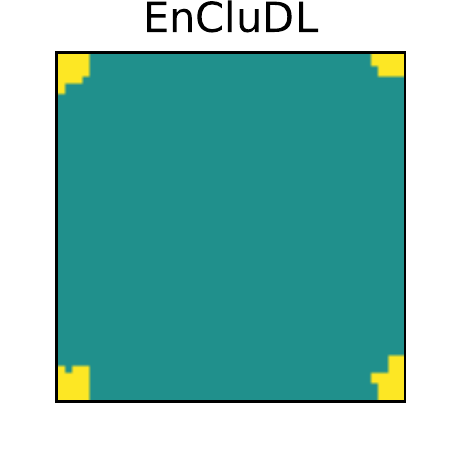}
      \end{minipage}
      \caption{
      True support and estimated support for the first seed of the central scenario.
      \underline{Left:}
      The support in yellow is composed of four regions of
      width $h = 4$ covariates.
      The tolerance region in green surrounds the support,
      its width is $\delta = 6$ covariates.
      The remaining covariates in blue form the $\delta$-null region.
      \underline{Others:}
      The yellow squares are the covariates selected by each method.
      Knockoffs selects few covariates when
      controlling the $k$-FWER at $10\%$ for $k = 4$.
      Desparsified Lasso only retrieves $3$ covariates
      when controlling the FWER at $10\%$.
      For $C = 200$, CluDL and EnCluDL have good power and control
      the $\delta$-FWER at $10\%$ for $\delta = 6$.
      }
      \label{fig:estimated_weights}
\end{figure}

In \Cref{fig:estimated_weights}, we plot the maps estimated by
knockoffs, desparsified Lasso, CluDL and EnCluDL for $C = 200$
when solving the first seed of the central scenario simulation.
Regarding knockoffs and desparsified Lasso solutions,
we notice that the power is low and the methods select few covariates
in each predictive region.
The CluDL method is more powerful and recovers groups of covariates that
correspond more closely to the true weights.
However, the shape of the CluDL solution depends on
the clustering choice.
The EnCluDL solution seems even more powerful than the CluDL one
and recovers groups of covariates that correspond almost perfectly
to the true weights.
Both CluDL and EnCluDL are only accurate up to the spatial
tolerance which is $\delta = 6$, but EnCluDL fits the ground
truth more tightly.

\begin{figure}[!ht]
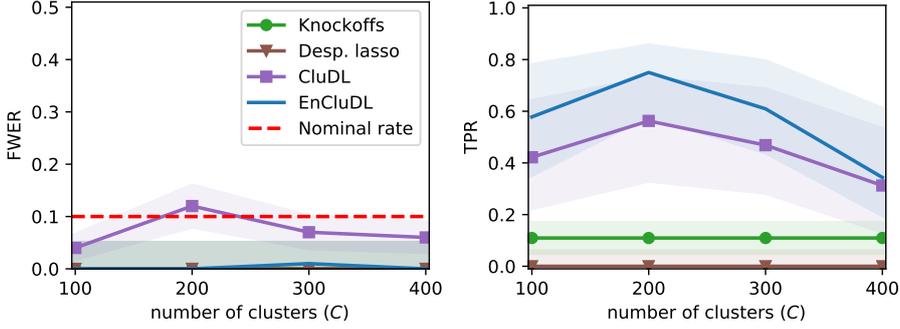

  \centering
  \begin{minipage}{0.4\linewidth}
    \centering\includegraphics[width=\linewidth]
    {fwer_main_scenario.pdf}
  \end{minipage}
  \begin{minipage}{0.4\linewidth}
    \centering\includegraphics[width=\linewidth]
    {tpr_main_scenario.pdf}
  \end{minipage}
  \caption{
  Results for fixed simulation parameters
  corresponding to the central scenario simulation.
  The green line with circles correspond to knockoffs,
  the brown line with triangles is the desparsified Lasso,
  the purple squared line correspond to CluDL and
  the blue plain line is EnCluDL.
  \underline{Left:} Empirical FWER for desparsified Lasso,
  $k$-FWER for knockoffs and $\delta$-FWER for CluDL and EnCluDL.
  The $80\%$ confidence intervals are obtained by Binomial approximation.
  \underline{Right:} Median true positive rate (TPR)
  for all the procedures, together with $80\%$ confidence interval
  obtained by taking the first decile and last decile TPR.
  }
  \label{fig:central_scenario}
\end{figure}

In \Cref{fig:central_scenario}, we focus on the central scenario to
get more insight about the statistical properties of the methods
and the influence of the $C$ hyper-parameter for CluDL and EnCluDL.
First, we observe that all methods reach the targeted control:
desparsified Lasso controls the FWER, knockoffs control the $k$-FWER
and, CluDL and EnCluDL control the $\delta$-FWER.
Second, considering the true positive rates (TPR), we notice that the
methods that do not integrate a compression step, \ie knockoffs
and desparsified Lasso, have a limited statistical power due to $n \ll p$.
However, CluDL has decent power and EnCluDL improves over CluDL thanks
to clustering randomization.
Finally, CluDL and EnCluDL are flexible with respect
to the choice of $C$ since the TPR varies quite slowly with $C$.

We have also studied the influence of the simulation parameters
by varying one parameter of the central scenario. The corresponding
results are available in \Cref{sec:additional_figures}.
The main conclusion gained from these complementary results is the fact
that, up to the limit given by the desired spatial tolerance $\delta$,
the choice of $C$ should be made in function of the data structure.
More precisely, good clustering creates clusters that are weakly
correlated and contains covariates that are highly correlated.
This observation is linked to assumption ($ii$) of
\Cref{prop:weights_compressed}.

\section{Discussion}
\label{sec:discussion}
%
When $n \ll p$, statistical inference on predictive model
parameters is a hard problem.
However, when the data are spatially structured, we have shown that
ensembled clustered inference procedures are attractive,
as they exhibit statistical guarantees and good power.
The price to pay is to accept that inference is only accurate up to
spatial distance $\delta$ corresponding to the clustering diameter,
thus replacing FWER with $\delta$-FWER control guarantees.

One of the most obvious field of application of this class of
algorithms is neuroscience where it can be used to solve source
localization problems.
In that regards, a wide empirical validation of EnCluDL has been
conducted in \citet{chevalier2021decoding} including fMRI data experiments.
Also, an extension of EnCluDL was proposed in
\citet{chevalier2020meg} to address the
magneto/electroencephalography source localization problem
which involves spatio-temporal data.

With EnCluDL, the statistical inference step is performed by
the desparsified Lasso.
In \citet{nguyen2019}, another ensembled clustered inference
method that leverages the knockoff technique
\citep{barber2015} leading to a procedure called ECKO
has been tested.
However, formal $\delta$-FDR control guarantees have not been
established yet for this model.
It would be also quite natural to try other inference techniques such as
the (distilled) conditional randomization test \citep{candes2018, liu2020fast}.

In the present work, we have only considered the linear regression setup.
However, combining the same algorithmic scheme with
statistical inference solutions for generalized linear models,
we could extend this work to the logistic regression setup.
This would extend the usability of ensembled clustered inference
to many more application settings.

\section*{Acknowledgement}
%
This study has been funded by Labex DigiCosme
(ANR-11-LABEX-0045-DIGICOSME) as part of the program
"Investissement d'Avenir" (ANR-11-IDEX-0003-02),
by the Fast-Big project (ANR-17-CE23-0011)
and the KARAIB AI Chair (ANR-20-CHIA-0025-01).
This study has also been supported by the European Union’s Horizon
2020 research and innovation program (Grant Agreement No. 945539,
Human Brain Project SGA3).

\section*{Supplementary material}
%
Supplementary material available online includes
an analysis of the technical assumptions and refinements
that occur when choosing the desparsified Lasso to perform the
statistical inference step in \Cref{sec:stat_inf_dl},
a diagram summarizing EnCluDL and a study
of the complexity of EnCluDL in \Cref{sec:algorithm_ECDL_journal},
a proposition for relaxing assumption $(ii)$ of
\Cref{prop:weights_compressed} in \Cref{sec:cdl_properties_corr},
complementary results for studying the influence of the simulation
parameters in \Cref{sec:additional_figures} and the proofs in
\Cref{sec:proofs_ECDL_journal}.

\clearpage

\bibliographystyle{plainnat}
\bibliography{./biblio}

\clearpage

\appendixpageoff
\appendixtitleoff
\renewcommand{\appendixtocname}{Supplementary material}
\begin{appendices}
{\LARGE{\textbf{Supplementary material for ``Spatially relaxed inference on high-dimensional linear models"}}}
\crefalias{section}{supp}
\crefalias{subsection}{supp}

\section{Desparsified Lasso on the compressed model}
\label{sec:stat_inf_dl}

Here, we clarify the assumptions and refinements that occur
when chosing the desparsified Lasso as the procedure that
performs the statistical inference on the compressed model.
The desparsified Lasso was first developed in \citet{Zhang_Zhang14}
and \citet{Javanmard_Montanari14}, and thoroughly analyzed in
\citet{vandeGeer_Buhlmann_Ritov_Dezeure14}.
Following notation in Eq.~\Cref{eq:noise_model_compressed},
the true support in the compressed model is denoted by
$S(\bm\theta^*) = \discsetin{c \in [C] : \bm\theta^*_c \neq 0}$
and its cardinality by $s(\bm\theta^*) = |S(\bm\theta^*)|$.
We also denote by $\bm\Omega \in \bbR^{C \times C}$ the inverse
of the population covariance matrix of
the groups, \ie $\bm\Omega = \bm\Upsilon^{-1}$.
Then, for $c \in [C]$, the sparsity of
the $c$-th row of $\bm\Omega$ (or $c$-th column) is
$s(\bm\Omega_{c,.}) = |S(\bm\Omega_{c,.})|$, where
$S(\bm\Omega_{c,.}) = \{c^\prime \in [C] : \bm\Omega_{c,c^\prime} \neq 0\}$.
We also denote the smallest eigenvalue of $\bm\Upsilon$ by
$\phi_{\min}(\bm\Upsilon) > 0$.
We can now state the assumptions required for
probabilistic inference with desparsified Lasso
\citep{vandeGeer_Buhlmann_Ritov_Dezeure14}:
\begin{thm}[Theorem 2.2 of \citet{vandeGeer_Buhlmann_Ritov_Dezeure14}]
\label{prop:Desparsified_Lasso}
Considering the model in Eq.~\Cref{eq:noise_model_compressed} and assuming:
\begin{align*}
\begin{split}
& (i) ~ 1 / \phi_{\min}(\bm\Upsilon) = \mathcal{O}(1) \enspace , \\
& (ii) ~ \max_{c \in [C]}(\bm\Upsilon_{c,c}) = \mathcal{O}(1) \enspace , \\
& (iii) ~ s(\theta^*) = o(\sqrt{n} / \log(C)) \enspace , \\
& (iv) ~ \max_{c \in [C]}(s(\bm\Omega_{c,.})) = o(n / \log(C)) \enspace , \\
\end{split}
\end{align*}
then, denoting by $\hat{\bm\theta}$ the desparsified Lasso estimator
derived from the inference procedure described in
\citet{vandeGeer_Buhlmann_Ritov_Dezeure14}, the following holds:
\begin{align*}
\begin{split}
& \sqrt{n}(\hat{\bm\theta} - \bm\theta^*) = \bm\xi + \bm\zeta \enspace , \\
& \bm\xi | \*Z \sim \mathcal{N}(0_{C}, \sigma_{\bm\eta}^2 \hat{\bm\Omega})
\enspace , \\
& \norm{\bm\zeta}_{\infty} = o_{\bbP}(1) \enspace , \\
\end{split}
\end{align*}
where $\hat{\bm\Omega}$ is such that
$\norm{\hat{\bm\Omega} - \bm{\Omega}}_{\infty} = o_{\bbP}(1)$.
\end{thm}
\begin{rk}\label{rk:noise_estimation}
  In \Cref{prop:Desparsified_Lasso}, to compute confidence intervals,
  the noise standard deviation $\sigma_{\eta}$ in the compressed problem
  has to be estimated.
  We refer the reader to the surveys
  that are dedicated to this subject
  such as \citet{reid2016, ndiaye2017, yu2019estimating}.
\end{rk}

As argued in \citet{vandeGeer_Buhlmann_Ritov_Dezeure14}, from
\Cref{prop:Desparsified_Lasso} we obtain asymptotic confidence
intervals for the $r$-th element of $\bm\theta^*$ from the following
equations, for all $z_1 \in \bbR$ and $z_2 \in \bbR^+$:
\begin{align}\label{eq:confint}
\begin{split}
& \bbP \left[ \frac{\sqrt{n}(\hat{\bm\theta}_c - \bm\theta_c^*)}
{\sigma_{\eta} \sqrt{\hat{\bm\Omega}_{c,c}}}
\leq z_1 ~\bigg|~ \*Z \right] - \Phi(z_1) = o_{\bbP}(1) \enspace , \\
& \bbP \left[ \frac{\sqrt{n}|\hat{\bm\theta}_c - \bm\theta_c^*|}
{\sigma_{\eta} \sqrt{\hat{\bm\Omega}_{c,c}}}
\leq z_2 ~\bigg|~ \*Z \right] - (2\Phi(z_2) - 1) = o_{\bbP}(1) \enspace , \\
\end{split}
\end{align}
where $\Phi(\cdot)$ is the cumulative distribution function of
the standard normal distribution.
Thus, for each $c \in [C]$ one can provide a p-value that assesses
whether or not $\bm\theta^*_c$ is equal to zero.
In the case of a two-sided single test, for each $c \in [C]$,
the p-value denoted by $\hat{p}^{\mathcal{G}}_c$ is:
\begin{align}\label{eq:p_value}
  \hat{p}^{\mathcal{G}}_c = 2 \left(1 -
  \Phi \left( \frac{\sqrt{n} |\hat{\bm\theta}_c|}
  {\sigma_{\eta} \sqrt{\hat{\bm\Omega}_{c,c}}} \right) \right)
  \enspace .
\end{align}
Under $H_{0}(G_c)$, from \Cref{eq:confint}, we have,
for any $\alpha \in (0, 1)$:
\begin{align}\label{eq:p_value_property_compressed_model}
\begin{split}
  \bbP(\hat{p}^{\mathcal{G}}_c \leq \alpha ~|~ \*Z)
  & = 1 - \bbP\left[ \frac{\sqrt{n} |\hat{\bm\theta}_c|}
  {\sigma_{\eta} \sqrt{\hat{\bm\Omega}_{c,c}}}
  \leq \Phi^{-1} \left( 1 - \frac{\alpha}{2} \right) ~\bigg|~ \*Z \right] \\
  & = \alpha + o_{\bbP}(1) \enspace .
\end{split}
\end{align}
Then, \Cref{eq:p_value_property_compressed_model} shows that the p-values
$\hat{p}^{\mathcal{G}}_c$ asymptotically control type 1 errors.
Using the Bonferroni correction, the family of corrected p-values
$\hat{q}^{\mathcal{G}} = (\hat{q}^{\mathcal{G}}_c)_{c \in [C]}$
remains defined by:
\begin{align}\label{eq:corrected_p_value_dl}
\hat{q}^{\mathcal{G}}_c = \min\{1, C \times {\hat{p}^{\mathcal{G}}_c}\}
\enspace .
\end{align}
Then, for all $\alpha \in (0, 1)$:
\begin{align}\label{eq:control_fwer_2_dl}
\mbox{FWER}_{\alpha}(\hat{q}^{\mathcal{G}}) =
\bbP(\min_{c \in N_{\mathcal{G}}}\hat{q}^{\mathcal{G}}_c \leq \alpha ~|~ \*Z)
\leq \alpha + o_{\bbP}(1) \enspace .
\end{align}
Then, \Cref{eq:control_fwer_2_dl} shows that the p-value family
$\hat{q}^{\mathcal{G}}$ asymptotically control FWER.
Finally, we have shown that desparsified Lasso
applied to a compressed version of the original problem
provides cluster-wise p-value families $\hat{p}^{\mathcal{G}}$
and $\hat{q}^{\mathcal{G}}$ that control respectively the type 1
error and the FWER in the compressed model only asymptotically.

\section{Relaxing the uncorrelated clusters assumption}
\label{sec:cdl_properties_corr}
%
As noted in \Cref{sec:compressed_representation_property},
assumption $(ii)$ of \Cref{prop:weights_compressed} is often
unmet in practice.
Here, taking the particular case in which the inference step is performed by
desparsified Lasso, we relax the assumption and show that it is still
possible to compute an adjusted corrected p-value that asymptotically controls
the $\delta$-FWER.
Hopefully, the technique used to derive this relaxation would also applicable
to other parametric statistical inference methods such as corrected ridge.
To better understand the development made in this section,
the adjusted p-values of this section should be compared with
the original p-values of \Cref{sec:stat_inf_dl}.
Note that, this extension is easy to integrate in the proof of the main results
\Cref{prop:ecdl_properties_final} as it just requires to use the
adjusted corrected p-value instead of the original corrected p-value.
Also, it does not provide much more insight about clustered inference
algorithms.
This is why we have decided to keep this extension for Supplementary Materials.

First, we replace \Cref{prop:weights_compressed} by the next proposition
that is a consequence of \citet[Proposition 4.4]{Buhlmann2012}.
\begin{prop}\label{prop:weights_compressed_2}
Considering the Gaussian linear model in \Cref{eq:noise_model_theory}
and assuming:
  \begin{enumerate}[leftmargin=*]
    \item[(i)] for all $c \in [C]$, for all  $j, k \in G^2_c, ~
    \Cov(\*X_{.,j}, \*X_{.,k} ~|~ \{\*Z_{.,c^\prime} : c^\prime \neq c \}) \geq 0$ \enspace ,
    \item[(ii.a)] {for all}~ $c \in [C], ~ \text{there exists} ~ \bm\nu_c \in \bbR^{+} \text{ s.t. }
    ~ \text{for all}~ j \in G_c, ~ \text{for all}~ k \notin G_c$~,
  \begin{align*}
  |\Cov(\*X_{.,j}, \*X_{.,k} ~|~ \{\*Z_{.,c^\prime} : c^\prime \neq c \})| \leq \bm\nu_c \enspace ,
  \end{align*}
    \item[(ii.b)] ~ {for all}~ $c \in [C], ~ \text{there exists} ~ \bm\tau_c > 0 \text{ s.t. }
    \Var(\*Z_{.,c}~|~\{\*Z_{.,c^\prime} : c^\prime \neq c \}) \geq \bm\tau_c$ \enspace ,
    \item[(iii)] ~ {for all}~ $c \in [C], ~
    \left(\text{for all}~ j \in G_c, \bm\beta^*_j \geq 0 \right) \orr
    \left(\text{for all}~ j \in G_c, \bm\beta^*_j \leq 0 \right)$ \enspace ,
  \end{enumerate}
then, in the compressed representation \Cref{eq:noise_model_compressed},
$\bm\theta^*$ admits the following decomposition:
\begin{align}\label{eq:weights_compressed_2}
  \bm\theta^* = \tilde{\bm\theta} + \bm\kappa \enspace ,
\end{align}
where, for all $c \in [C]$,
$|\bm\kappa_c| \leq (\bm\nu_c ~ / \bm\tau_c) \normin{\bm\beta^*}_{1}$
and $\tilde{\bm\theta}_c \neq 0$ if and only if
there exists $j \in G_c$ such that $\bm\beta_j^* \neq 0$.
If such an index $j$ exists then
$\sign(\tilde{\bm\theta}_c) = \sign(\bm\beta_j^*)$.
\end{prop}
\begin{proof}
See \Cref{sec:proof_weights_compressed}.
\end{proof}
The assumptions $(i)$ and $(ii)$ in \Cref{prop:weights_compressed}
are replaced by $(i)$, $(ii.a)$ and $(ii.b)$
in \Cref{prop:weights_compressed_2}.
More precisely, instead of assuming that the covariates
inside a group are positively correlated, we assume that they are positively
correlated conditionally to all other groups.
Also, we relax the more questionable assumption of groups independence;
we assume instead that the conditional covariance of two covariates of
different groups is bounded above $(ii.a)$ and that the conditional variance
of the group representative variable is non-zero $(ii.b)$.
In practice, except when group representative variables are linearly dependent,
we can always find values for which $(ii.a)$ and $(ii.b)$ are verified,
but we would like the upper bound of $(ii.a)$ as low as possible and
the lower bound of $(ii.b)$ as high as possible.
Finally, assumption $(iii)$ remains unchanged.

Then, as done in \Cref{sec:stat_inf_dl}, we can build $\hat{\bm\theta}$.
Under the same assumptions, \Cref{prop:Desparsified_Lasso} is still valid
and $\hat{\bm\theta}$ still verifies \Cref{eq:confint}.
However, here we want to estimate $\tilde{\bm\theta}$, not $\bm\theta^*$.
Combining \Cref{prop:Desparsified_Lasso} and \Cref{prop:weights_compressed_2},
we can see $\hat{\bm\theta}$ as a biased estimator of $\tilde{\bm\theta}$.
To take this bias into account, we need to adjust the definition of the
p-values given by \Cref{eq:p_value}.
Let us assume that, for a given $a \in \bbR^{+}$,
\begin{align}\label{eq:p_value_assumption}
\max_{c \in [C]} \left( \frac{\bm\nu_c} {\bm\tau_c \sqrt{\hat{\bm\Omega}_{c,c}}} \right)
\leq \frac{a \, \sigma_{\varepsilon}} {\norm{\bm\beta^*}_{1}}
\enspace .
\end{align}
And, for all $c \in [C]$, let us define the adjusted p-values:
\begin{align}\label{eq:p_value_2}
\hat{p}^{\mathcal{G}}_c =
  2 \left( 1 - \Phi
    \left(
    \sqrt{n} \left[
    \frac{|\hat{\bm\theta}_c|}{\sigma_{\eta} \sqrt{\hat{\bm\Omega}_{c,c}}} - a
    \right]_{+}
    \right)
  \right)
\enspace .
\end{align}
Let us denote by $q_{1 - \frac{\alpha}{2}}=\Phi^{-1} (1 - \frac{\alpha}{2})$
the $1 - \frac{\alpha}{2}$ quantile of the standard Gaussian distribution.
Then, under $H_{0}(G_c)$, the hypothesis which states that
$\bm\beta^*_j = 0$ for $j \in G_c$ implying that $\tilde{\bm\theta}_c = 0$,
we have, for any $\alpha \in (0, 1)$:

\begin{align}\label{eq:adjusted_p_value_property_compressed_model}
  \begin{split}
    \bbP(\hat{p}^{\mathcal{G}}_c \leq \alpha ~|~ \*Z)
    & = 1 - \bbP\left[ \sqrt{n} \left[
    \frac{|\hat{\bm\theta}_c|}{\sigma_{\eta} \sqrt{\hat{\bm\Omega}_{c,c}}} - a
                                \right]_{+}
    \leq q_{1 - \frac{\alpha}{2}} ~\bigg|~ \*Z \right] \\
    & \leq 1 - \bbP\left[ \sqrt{n} \left[
    \frac{|\hat{\bm\theta}_c|}{\sigma_{\eta} \sqrt{\hat{\bm\Omega}_{c,c}}} -
    \frac{\bm\nu_c \norm{\bm\beta^*}_{1}}
        {\sigma_{\varepsilon} \bm\tau_c \sqrt{\hat{\bm\Omega}_{c,c}}}
                                \right]_{+}
    \leq q_{1 - \frac{\alpha}{2}} ~\bigg|~ \*Z \right] \\
    & \leq 1 - \bbP\left[ \sqrt{n} \left[
    \frac{|\hat{\bm\theta}_c| - |\bm\kappa_c|}
    {\sigma_{\eta} \sqrt{\hat{\bm\Omega}_{c,c}}}
                                \right]_{+}
    \leq q_{1 - \frac{\alpha}{2}} ~\bigg|~ \*Z \right] \\
    & = 1 - \bbP\left[ \sqrt{n} \left[
    \frac{|\hat{\bm\theta}_c| - |\bm\theta^{*}_c|}
    {\sigma_{\eta} \sqrt{\hat{\bm\Omega}_{c,c}}}
                                \right]_{+}
    \leq q_{1 - \frac{\alpha}{2}} ~\bigg|~ \*Z \right] \\
    & \leq 1 - \bbP\left[ \sqrt{n}
    \frac{|\hat{\bm\theta}_c - \bm\theta^{*}_c|}
    {\sigma_{\eta} \sqrt{\hat{\bm\Omega}_{c,c}}}
    \leq q_{1 - \frac{\alpha}{2}} ~\bigg|~ \*Z \right] \\
    & = \alpha + o_{\bbP}(1) \enspace .
  \end{split}
\end{align}
Finally, we have built a cluster-wise adjusted p-value family that
asymptotically exhibits, with low probability ($< \alpha$),
low value ($< \alpha$) for the clusters
which contain only zero weight covariates.
To complete the proof in the case of correlated clusters,
one can proceed as in uncorrelated cluster case
taking \Cref{eq:p_value_2} instead of \Cref{eq:p_value}.

Now, let us come back to the interpretation and choice for the constant $a$.
In \Cref{prop:weights_compressed_2}, we have shown that, when groups
are not independent, a group weight in the compressed model can be
non-zero even if the group only contains zero weight covariates.
However, the absolute value of the weight of such a group is necessarily upper
bounded.
We thus introduce $a \in \bbR^{+}$ in \Cref{eq:p_value_2} to
increase the p-values by a relevant amount and keep
statistical guarantees concerning the non-discovery of a such group.
The value of $a$ depends on the physics of the problem
and on the choice of clustering.
While the physics of the problem is fixed, the choice of clustering
has a strong impact on the left term of \Cref{eq:p_value_assumption} and
a "good" choice of clustering results in a lower $a$ (less correction).
To estimate $a$, we need to find an upper bound of $\norm{\bm\beta^*}_{1}$,
a lower bound of $\sigma_{\varepsilon}$ and to estimate the left term of
\Cref{eq:p_value_assumption}.
In practice, to compute p-values, we took $a = 0$ since
the formula in \Cref{eq:p_value} was already conservative
for all the problems we considered.

\section{EnCluDL}
\label{sec:algorithm_ECDL_journal}
%
\begin{figure}[!ht]
    \centering
    \includegraphics[width=1.0\linewidth]{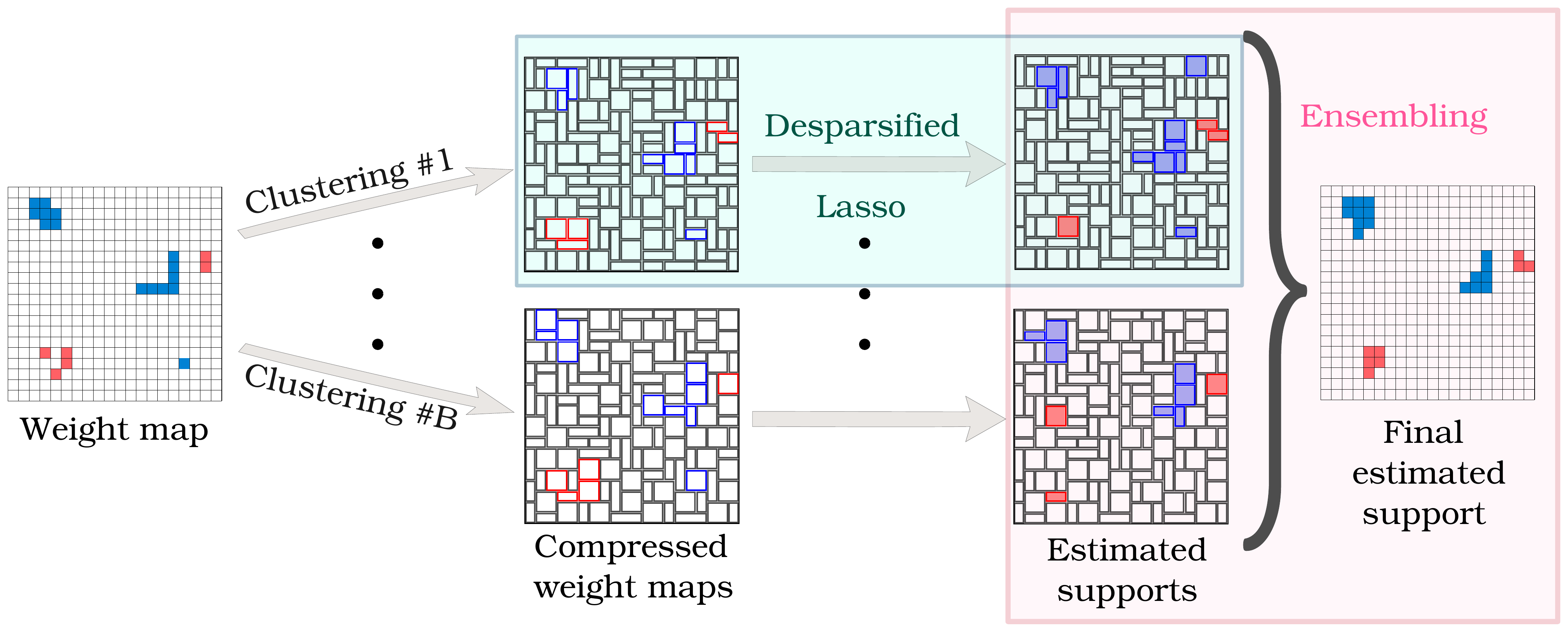}
    \caption{Summary of the mechanism of ensemble of clustered
             desparsified Lasso (EnCluDL).
             EnCluDL combines three algorithmic steps: a clustering
             procedure, the desparsified Lasso statistical inference procedure
             to derive p-value maps, and an ensembling method that
             synthesizes several p-value maps into one.
    }
    \label{fig:ECDL}
\end{figure}

Computationally, to derive the EnCluDL solution we must solve
$B$ independent CluDL problems, making the global problem
embarrassingly parallel; nevertheless, we could run the CluDL
algorithm on standard desktop stations without parallelization
with $n = 400$, $p \approx 10^5$, $C = 500$ and $B = 25$ in
less than $10$ minutes.
Note that, the clustering step being much quicker than the inference
step, $p$ has a very limited impact on the total computation time.

The complexity for solving the Lasso depends significantly on the
choice of solver, we then give the complexity in numbers of Lasso.
The complexity for solving EnCluDL is given by the complexity of the
resolution of $\mathcal{O}(B \times C)$ Lasso problems with $n$ samples and
$C$ covariates, \ie with clustering.
It is noteworthy that the complexity for solving the desparsified Lasso on the
original problem is given by the complexity of the resolution of
$\mathcal{O}(p)$ Lasso problems with $n$ samples and $p$ covariates, \ie
without clustering.
Then, EnCluDL should be much faster than the
desparsified Lasso whenever $p \gg C$.

\section{Complementary simulation results}
\label{sec:additional_figures}
%
\begin{figure}[!ht]
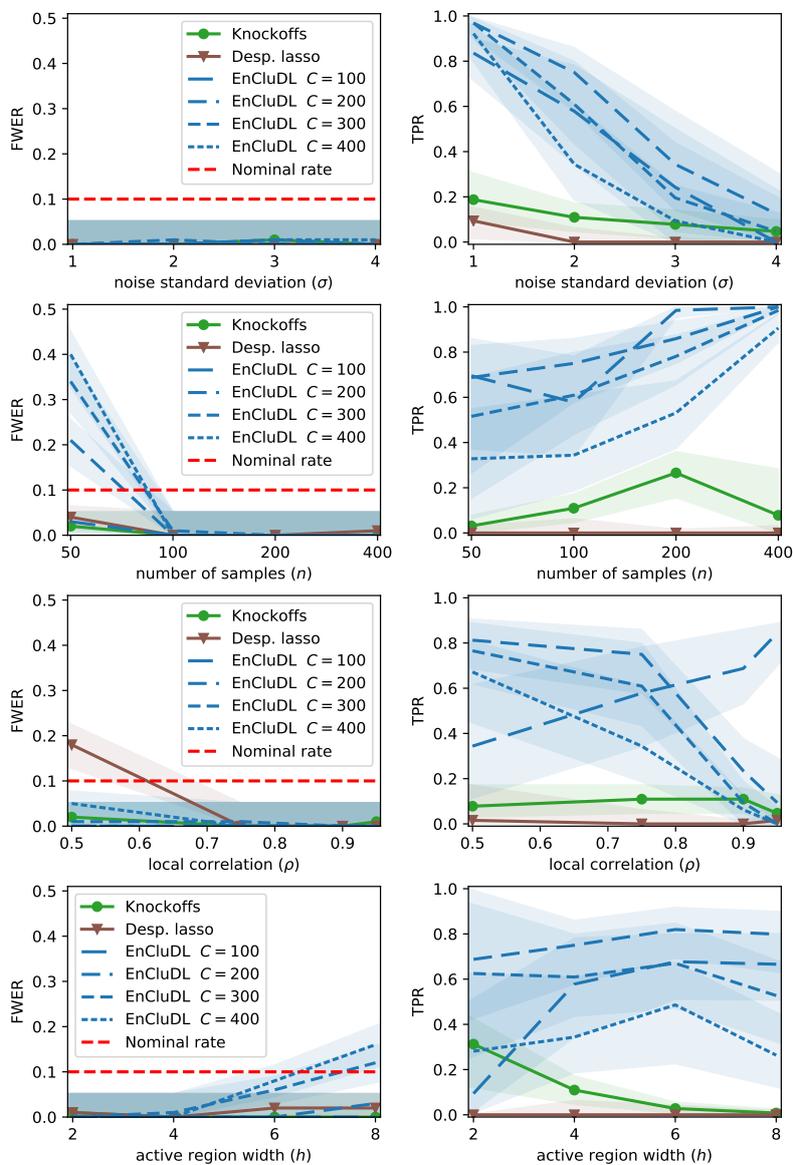

    \centering
    \begin{minipage}{0.35\linewidth}
      \centering\includegraphics[width=\linewidth]{fwer_sigma.pdf}
    \end{minipage}
    \begin{minipage}{0.35\linewidth}
      \centering\includegraphics[width=\linewidth]{tpr_sigma.pdf}
    \end{minipage}
    \begin{minipage}{0.35\linewidth}
      \centering\includegraphics[width=\linewidth]{fwer_samples.pdf}
    \end{minipage}
    \begin{minipage}{0.35\linewidth}
      \centering\includegraphics[width=\linewidth]{tpr_samples.pdf}
    \end{minipage}
    \begin{minipage}{0.35\linewidth}
      \centering\includegraphics[width=\linewidth]{fwer_smoothing.pdf}
    \end{minipage}
    \begin{minipage}{0.35\linewidth}
      \centering\includegraphics[width=\linewidth]{tpr_smoothing.pdf}
    \end{minipage}
    \begin{minipage}{0.35\linewidth}
      \centering\includegraphics[width=\linewidth]{fwer_roi_size.pdf}
    \end{minipage}
    \begin{minipage}{0.35\linewidth}
      \centering\includegraphics[width=\linewidth]{tpr_roi_size.pdf}
    \end{minipage}
    \caption{
    Results for various simulation parameters.
    The green line with circles correspond to the knockoffs,
    the brown line with triangle is the desparsified lasso,
    the dashed blue lines are for EnCluDL with length of the dashes
    increasing when $C$ diminishes: large dashes are for $C = 100$,
    medium for $C = 200$, small for $C = 300$, tiny for $C = 400$.
    We compute the same FWER and TPR quantities as in
    \Cref{fig:central_scenario}, and the same $80\%$ confidence intervals:
    by Binomial approximation for the FWER and taking first and last
    deciles for the TPR.
    }
    \label{fig:all_scenarios}
\end{figure}

In \Cref{fig:all_scenarios}, we study the influence of the simulation
parameters by varying one parameter of the central scenario at a time.
We vary the noise standard deviation, the number of samples, the
local correlation and the size of the support.
For a better readability of the figures, we do not analyze the results
of CLuDL since it is expected to be always a bit less powerful than
EnCluDL while showing a similar behavior.
First, we look at the plots where we vary the noise standard deviation
$\sigma$.
We observe that the methods reach the targeted FWER control and notice that
EnCluDL benefits more strongly from the decrease of
$\sigma$ regarding support recovery.
Second, we analyze the results for various sample sizes ($n$) values.
Concerning EnCluDL, we notice that the $\delta$-FWER is not controlled
when $n = 50$ except for $C = 100$.
This is not surprising since the $\delta$-FWER control is asymptotic and
$n = 50$ is not sufficient.
In terms of support recovery, the problem gets easier with larger $n$, but
only EnCluDL benefits strongly from an increase of $n$.
Third, we investigate the influence of the level of correlation between
neighboring covariates ($\rho$).
Regarding FWER control, desparsified lasso does not control the FWER when
$\rho = 0.5$.
Regarding the statistical power of EnCluDL, as one would expect,
when the spatial structure is strong \ie $\rho > 0.9$, it is relevant
to pick larger clusters, \ie to take a smaller $C$.
Indeed, to make a relevant choice for $C$, data structure
has to be taken into account to derive good covariates' clustering;
this is true up to the limit given by the desired spatial tolerance.
A good clustering creates clusters that are weakly correlated and
contains covariates that are highly correlated.
This observation is linked to assumption ($ii$) of
\Cref{prop:weights_compressed} or to assumption ($ii.a$) and ($ii.b$)
of \Cref{prop:weights_compressed_2}.
Finally, we consider the results for different support sizes coded by
the active region width $h$.
Sparsity is a crucial assumption for desparsified lasso and
then for EnCluDL.
Also, when $p$ (or $C$) increases the required sparsity is greater.
This explains why when $h = 8$ and $C \geq 300$, the empirical
$\delta$-FWER is slightly above the expected nominal rate.
Regarding the statistical power of EnCluDL, as one could expect,
when the active regions are large, it is relevant to use large clusters.
However, it can be difficult to estimate this parameter in advance, thus
we prefer to consider desired spatial tolerance parameter $\delta$ and
data structure to set $C$.

\section{Proofs}
\label{sec:proofs_ECDL_journal}

\subsection{Proof of \Cref{prop:weights_compressed} and \Cref{prop:weights_compressed_2}}
\label{sec:proof_weights_compressed}

First, we start by the proof of \Cref{prop:weights_compressed} which
is derived from \citet[Proposition 4.3]{Buhlmann2012}:

\begin{proof}

With assumption $(ii)$ and \citet[Proposition 4.3]{Buhlmann2012},
we have, for all $c \in [C]$:
\begin{align*}
  \bm\theta_c^* = |G_c| \sum_{j \in G_c} w_j \bm\beta^*_j \enspace ,
\end{align*}
where, for all $j\in G_c$:
\begin{align*}
  w_j = \frac{\sum_{k \in G_c} \bm\Sigma_{j,k}}
  {\sum_{k \in G_c}\sum_{k^\prime \in G_c} \bm\Sigma_{k,k^\prime}} \enspace .
\end{align*}
From assumption $(i)$, we have $w_j > 0$ for all $j \in G_c$.
Assumption $(iii)$ ensures that, for all $j \in G_c$, the $\bm\beta^*_j$
have the same sign.
Then, $\bm\theta_c^*$ is of the same sign as the $\bm\beta^*_j$
and is non-zero only if there exists $j \in G_c$ such that
$\bm\beta^*_j \neq 0$.
\end{proof}

Now, we give the proof of \Cref{prop:weights_compressed_2} which
is mainly derived from \citet[Proposition 4.4]{Buhlmann2012}:

\begin{proof}
With assumption $(ii.a)$ and $(ii.b)$ and
\citet[Proposition 4.4]{Buhlmann2012},
we have, for all $c \in [C]$:
\begin{align*}
  \bm\theta_c^* = |G_c| \sum_{j \in G_c} w^{\prime}_j \bm\beta^*_j
                + \bm\kappa_c \enspace ,
\end{align*}
where
\begin{align*}
  w^{\prime}_j =
  \frac{\sum_{k \in G_c} \Cov(\*X_{.,j}, \*X_{.,k}~|~\{\*Z_{.,c^\prime} : c^\prime \neq c \})}
       {\sum_{k \in G_c} \sum_{k^\prime \in G_c}
        \Cov(\*X_{.,k}, \*X_{.,k^\prime}~|~\{\*Z_{.,c^\prime} : c^\prime \neq c \})}
  \enspace ,
\end{align*}
and, for all $c \in [C]$
\begin{align*}
  |\bm\kappa_c| \leq (\bm\nu_c ~ / \bm\tau_c) \normin{\bm\beta^*}_{1}
  \enspace .
\end{align*}
Let us define $\tilde{\bm\theta}$ by
\begin{align*}
  \tilde{\bm\theta}_c = |G_c| \sum_{j \in G_c} w^{\prime}_j \bm\beta^*_j
                      \enspace .
\end{align*}
Then,
\begin{align*}
  \bm\theta^* = \tilde{\bm\theta} + \bm\kappa \enspace ,
\end{align*}
And, similarly as in the proof of \Cref{prop:weights_compressed},
from assumption $(i)$ and $(iii)$, $\tilde{\bm\theta}_c$ is of the
same sign as the $\bm\beta^*_j$ for $j \in G_c$ and is non-zero only
if there exists $j \in G_c$ such that $\bm\beta^*_j \neq 0$.
\end{proof}

\subsection{Proof of \Cref{prop:p_values_control}}
\label{sec:proof_degrouping}
%
Before going trough the proof of \Cref{prop:p_values_control},
we introduce the grouping function $g$ that matches the
covariate index to its corresponding group index:
\begin{align*}
\begin{split}
g: [p] &\to [C] \\
~j~ & \mapsto ~c~ \quad \si  ~ j \in G_c  \enspace . \\
\end{split}
\end{align*}
Then, \Cref{eq:de_grouping_p_values} can be rewritten as follows:
\begin{align}\label{eq:de_grouping_p_values_2}
\begin{split}
& \text{for all}~ j \in [p], \quad \hat{p}_j =
\hat{p}^{\mathcal{G}}_{g(j)} \enspace , \\
& \text{for all}~ j \in [p], \quad \hat{q}_j =
\hat{q}^{\mathcal{G}}_{g(j)}\enspace .\\
\end{split}
\end{align}

\begin{proof}
\textit{(i)}
Suppose that we are under $H^{\delta}_0(j)$.
Since the cluster diameters are all smaller than $\delta$, all
the covariates in $G_{g(j)}$ have a corresponding weight equal to zero.
Thus, using \Cref{prop:weights_compressed}, we have
$\bm\theta_{g(j)}^* = 0$, \ie we are under $H_0(G_{g(j)})$.
Under this last null-hypothesis, using
\Cref{eq:p_value_property_compressed_model} and
\Cref{eq:de_grouping_p_values_2}, we have:
\begin{align*}
\text{for all}~ \alpha \in (0, 1),
~ \bbP(\hat{p}_{g(j)}^{\mathcal{G}} \leq \alpha) =
\bbP(\hat{p}_{j} \leq \alpha) = \alpha \enspace .
\end{align*}
This last result being true for any $j \in N^{\delta}$, we have
shown that the elements of the family $\hat{p}$
control the $\delta$-type 1 error.

\medskip

\textit{(ii)}
As mentioned in \Cref{sec:stat_inf}, we know that, the family
$\hat{q}^{\mathcal{G}}$ controls the FWER,
\ie for $\alpha \in (0,1)$ we have
$\bbP (\min_{c \in N_{\mathcal{G}}} \hat{q}^{\mathcal{G}}_c \leq \alpha)
\leq \alpha$.
Let us denote by $g^{-1}(N_{\mathcal{G}})$ the set of indexes of
covariates that belong to the groups of $N_{\mathcal{G}}$, \ie
$g^{-1}(N_{\mathcal{G}}) =
\left\{j \in [p] : {g(j)} \in N_{\mathcal{G}} \right\}$.
Again, given that all the cluster diameters
are smaller than $\delta$ and using
\Cref{prop:weights_compressed}, if $j \in N^{\delta}$ then
$g(j) \in N_{\mathcal{G}}$.
That is to say $N^{\delta} \subset g^{-1}(N_{\mathcal{G}})$.
Then, we have:
\begin{align*}
  \min_{j \in N^{\delta}} (\hat{q}_j) \geq
  \min_{j \in g^{-1}(N_{\mathcal{G}})} (\hat{q}_j) \enspace .
\end{align*}
We can also notice that:
\begin{align*}
\begin{split}
  \min_{j \in g^{-1}(N_{\mathcal{G}})} (\hat{q}_j)
  & = \min_{j \in g^{-1}(N_{\mathcal{G}})} (\hat{q}^{\mathcal{G}}_{g(j)}) \\
  & = \min_{g(j) \in N_{\mathcal{G}}} (\hat{q}^{\mathcal{G}}_{g(j)})
  \enspace .\\
\end{split}
\end{align*}
Replacing $g(j) \in [C]$ by $c \in [C]$, and using
\Cref{eq:control_fwer_2}, we obtain:
\begin{align*}
  \text{for all}~ \alpha \in (0,1),
  ~ \bbP (\min_{j \in N^{\delta}} (\hat{q}_j) \leq \alpha) \leq
  \bbP (\min_{c \in N_{\mathcal{G}}}\hat{q}^{\mathcal{G}}_c \leq \alpha)
  \leq \alpha \enspace .
\end{align*}
This last result states that the family $(\hat{q}_j)_{j \in [p]}$
controls the $\delta$-FWER.
\end{proof}

\subsection{Proof of \Cref{prop:p_values_aggregation}}
\label{sec:proof_aggregation}
%
The proof of \Cref{prop:p_values_aggregation} is inspired
by the one proposed by \citet{Meinshausen2008}.
However, it is subtly different since we can not remove the term
$\min_{j \in N^{\delta}}$ and have to work with it to obtained
the desired inequality.
First, we start by making a short remark about the $\gamma$-quantile quantity.
\begin{df}[empirical $\gamma\mbox{-quantile}$]\label{df:gamma-quantile}
For a set $V$  of real numbers and $\gamma \in (0, 1)$, let
\begin{align}\label{eq:gamma-quantile}
  \gamma\mbox{-quantile}(V) = \min \left\{ v \in V :
  \frac{1}{|V|} \sum_{w \in V} \mathds{1}_{w \leq v} \geq \gamma \right\}
  \enspace .
\end{align}
\end{df}
\begin{rk}\label{prop:gamma-quantile}
For a set of real number $V$ and for $a \in \bbR$, let us define
the quantity $\pi(a, V)$ by the following:
\begin{align}\label{eq:pi}
  \pi(a, V) = \frac{1}{|V|} \sum_{v \in V} \mathds{1}\left(v \leq a \right)
\end{align}
Then, for $\gamma \in (0, 1)$, the two events
$E_1 = \{ \pi(a, V) \geq \gamma \}$ and
$E_2 = \{ \gamma\mbox{-quantile}(V) \leq a \}$ are identical.
\end{rk}
Now, we give the proof of \Cref{prop:p_values_aggregation}.
\begin{proof}
First, one can notice that, from \Cref{eq:p_values_aggregation}, we have:
\begin{align*}
    \min_{j \in N^{\delta}} (\tilde{q}_j(\gamma))
    \geq \min \left\{ 1, \gamma\mbox{-quantile}
    \left( \left\{ \min_{j \in N^{\delta}}
    \left(\frac{\hat{q}^{(b)}_j}{\gamma} \right) :
    b \in [B] \right\} \right) \right\} \enspace .
\end{align*}
Then, for $\alpha \in (0, 1)$:
\begin{align*}
  \begin{split}
    \bbP\left(\min_{j \in N^{\delta}} (\tilde{q}_j(\gamma)) \leq \alpha \right)
    & \leq \bbP\left(\min \left\{ 1, \gamma\mbox{-quantile}
    \left( \left\{ \min_{j \in N^{\delta}}
    \left(\frac{\hat{q}^{(b)}_j}{\gamma} \right) :
    b \in [B] \right\}\right)\right\} \leq \alpha \right) \\
    & = \bbP\left(\gamma\mbox{-quantile}
    \left( \left\{ \min_{j \in N^{\delta}}
    \left(\frac{\hat{q}^{(b)}_j}{\gamma} \right) :
    b \in [B] \right\} \right) \leq \alpha \right) \enspace . \\
  \end{split}
\end{align*}
Using \Cref{prop:gamma-quantile}, for $\gamma \in (0,1)$, with:
\begin{align*}
V = \left\{ \min_{j \in N^{\delta}}
\left( \frac{\hat{q}^{(b)}_j}{\gamma} \right) :
b \in [B] \right\}
\quad \andd \quad
a = \alpha
\enspace ,
\end{align*}
and noticing that:
\begin{align*}
  \pi\left(\alpha, \left\{ \min_{j \in N^{\delta}}
  \left( \frac{\hat{q}^{(b)}_j}{\gamma} \right) :
  b \in [B] \right\}\right)
  =
  \frac{1}{B} \sum^{B}_{b = 1} \mathds{1}
  \left\{\min_{j \in N^{\delta}}({\hat{q}^{(b)}_j})
  \leq \alpha \gamma \right\} \enspace ,
\end{align*}
then, we have:
\begin{align*}
\bbP \left( \gamma\mbox{-quantile}
\left( \left\{ \min_{j \in N^{\delta}}
\left(\frac{\hat{q}^{(b)}_j}{\gamma} \right) :
b \in [B] \right\}\right) \leq \alpha \right)
=
\bbP \left( \frac{1}{B} \sum^{B}_{b = 1} \mathds{1}
\left\{\min_{j \in N^{\delta}}({\hat{q}^{(b)}_j}) \leq \alpha \gamma \right\}
\geq \gamma \right) .
\end{align*}
Then, the Markov inequality gives:
\begin{align*}
\bbP \left( \frac{1}{B} \sum^{B}_{b = 1} \mathds{1}
\left\{\min_{j \in N^{\delta}}({\hat{q}^{(b)}_j}) \leq \alpha \gamma \right\}
\geq \gamma \right)
\leq
\frac{1}{\gamma} \bbE \left[ \frac{1}{B} \sum^{B}_{b = 1} \mathds{1}
\left\{\min_{j \in N^{\delta}}({\hat{q}^{(b)}_j}) \leq \alpha \gamma \right\}
\right] \enspace .
\end{align*}
Then, using the assumption that the $B$ families
$(\hat{q}^{(b)}_j)_{j \in [p]}$ control of the
$\delta$-FWER (last inequality), we have:
\begin{align*}
\frac{1}{\gamma} \bbE \left[ \frac{1}{B} \sum^{B}_{b = 1} \mathds{1}
\left\{\min_{j \in N^{\delta}}({\hat{q}^{(b)}_j}) \leq \alpha \gamma \right\}
\right]
=
\frac{1}{\gamma} \frac{1}{B} \sum^{B}_{b = 1}
\bbP \left(\min_{j \in N^{\delta}}(\hat{q}^{(b)}_j) \leq \alpha \gamma \right)
\leq \alpha \enspace .
\end{align*}
Finally, we have shown that, for $\alpha \in (0, 1)$:
\begin{align*}
  \bbP\left(\min_{j \in N^{\delta}} (\tilde{q}_j(\gamma)) \leq \alpha \right)
  \leq \alpha \enspace .
\end{align*}
This establishes that the family $(\tilde{q}_j(\gamma))_{j \in [p]}$
controls the $\delta$-FWER.
\end{proof}
%
\subsection{Proof of \Cref{prop:ecdl_properties_final}}
\label{sec:proof_ecdl_properties_final}
%
To show \Cref{prop:ecdl_properties_final}, we connect the previous results:
\Cref{prop:weights_compressed}, \Cref{prop:p_values_control}
and \Cref{prop:p_values_aggregation}.
First, we prove that clustered inference algorithms produce
a p-value family that controls the $\delta$-FWER.
\begin{proof}
Assuming the noise model \Cref{eq:noise_model_theory}, assuming that
\Cref{ass:assumption_1} and \Cref{ass:assumption_2} are verified for a
distance parameter larger than $\delta$ and that the clustering
diameter is smaller than $\delta$, then we directly obtain the
assumption $(i)$ and $(iii)$ of \Cref{prop:weights_compressed}.
This means that the compressed representation has the correct pattern
of non-zero coefficients, in particular it does not include in the
support clusters of null-only covariates.
Additionally, if one is able to perform a valid statistical
inference on the compressed model \Cref{eq:noise_model_compressed},
\ie to produce cluster-wise p-values such that
\Cref{eq:p_value_type_1_error_control} holds,
then \Cref{prop:p_values_control} ensures that the p-value
family constructed using the de-grouping method presented in
\Cref{eq:de_grouping_p_values} controls the $\delta$-FWER.
\end{proof}

Now, we prove that ensembled clustered inference algorithms
produce a p-value family that controls the $\delta$-FWER.
\begin{proof}
  Given the above arguments, the p-value families produced by
  clustered inference algorithms subject to all clusterings
  fulfilling the theorem hypotheses control the $\delta$-FWER.
  Then, using the aggregation method given by \Cref{eq:p_values_aggregation},
  we know from \Cref{prop:p_values_aggregation} that the aggregated p-value
  family also controls the $\delta$-FWER.
\end{proof}
\end{appendices}
\end{document}